\documentclass[a4paper,UKenglish,cleveref, autoref]{lipics-v2019}
\usepackage{amsmath,amsfonts,amsthm,mathrsfs,mathpazo,xspace,hyperref,graphicx}
\usepackage{endnotes}
\usepackage{color}
\usepackage{bm}
\usepackage{amssymb,latexsym}
\usepackage{float}
\usepackage[ruled,vlined,linesnumbered]{algorithm2e}

\newtheorem{question}[theorem]{Question}

\newtheorem{oracle}{Oracle}
\theoremstyle{remark}

\numberwithin{equation}{section}
\numberwithin{oracle}{section}
\numberwithin{remark}{section}

\newcommand{\eqn}[1]{(\ref{eqn:#1})}
\newcommand{\ora}[1]{\hyperref[ora:#1]{Oracle~\ref*{ora:#1}}}
\newcommand{\thm}[1]{\hyperref[thm:#1]{Theorem~\ref*{thm:#1}}}
\newcommand{\cor}[1]{\hyperref[cor:#1]{Corollary~\ref*{cor:#1}}}
\newcommand{\defn}[1]{\hyperref[defn:#1]{Definition~\ref*{defn:#1}}}
\newcommand{\lem}[1]{\hyperref[lem:#1]{Lemma~\ref*{lem:#1}}}
\newcommand{\prop}[1]{\hyperref[prop:#1]{Proposition~\ref*{prop:#1}}}
\newcommand{\fig}[1]{\hyperref[fig:#1]{Figure~\ref*{fig:#1}}}
\newcommand{\tab}[1]{\hyperref[tab:#1]{Table~\ref*{tab:#1}}}
\newcommand{\algo}[1]{\hyperref[algo:#1]{Algorithm~\ref*{algo:#1}}}
\newcommand{\ques}[1]{\hyperref[ques:#1]{Question~\ref*{ques:#1}}}
\newcommand{\rem}[1]{\hyperref[rem:#1]{Remark~\ref*{rem:#1}}}
\renewcommand{\sec}[1]{\hyperref[sec:#1]{Section~\ref*{sec:#1}}}
\newcommand{\append}[1]{\hyperref[append:#1]{Appendix~\ref*{append:#1}}}

\def\>{\rangle}
\def\<{\langle}

\newcommand{\range}[1]{[#1]}

\newcommand{\ket}[1]{|#1\rangle}
\newcommand{\bra}[1]{\langle#1|}

\newcommand{\proj}[1]{\ket{#1}\bra{#1}}

\newcommand{\E}{\mathbf{E}}

\newcommand{\C}{\ensuremath{\mathbb{C}}}

\newcommand{\R}{\ensuremath{\mathbb{R}}}
\newcommand{\Z}{\ensuremath{\mathbb{Z}}}

\newcommand{\eps}{\varepsilon}

\DeclareMathOperator{\poly}{poly}

\DeclareMathOperator{\tr}{Tr}
\DeclareMathOperator{\Tr}{Tr}
\DeclareMathOperator{\rank}{rank}

\newcommand{\hd}[1]{\vspace{2mm} \noindent \textbf{#1}}
\def \eps {\epsilon}

\title{Quantum SDP Solvers: Large Speed-ups, Optimality, and Applications to Quantum Learning} %TODO Please add

\titlerunning{Quantum SDP Solver: Speed-ups, Optimality, Applications}%optional, please use if title is longer than one line

\author{Fernando G. S. L. Brand\~{a}o}{Institute of Quantum Information and Matter, California Institute of Technology, USA}{fgslbrandao@gmail.com}{}{}%TODO mandatory, please use full name; only 1 author per \author macro; first two parameters are mandatory, other parameters can be empty. Please provide at least the name of the affiliation and the country. The full address is optional

\author{Amir Kalev}{Joint Center for Quantum Information and Computer Science, University of Maryland, USA}{amirk@umd.edu}{}{}

\author{Tongyang Li}{Joint Center for Quantum Information and Computer Science, University of Maryland, USA}{tongyang@cs.umd.edu}{}{}

\author{Cedric Yen-Yu Lin}{Joint Center for Quantum Information and Computer Science, University of Maryland, USA}{cedriclin37@gmail.com}{}{}

\author{Krysta M. Svore}{Station Q, Quantum Architectures and Computation Group, Microsoft Research, USA}{ksvore@microsoft.com}{}{}

\author{Xiaodi Wu}{Joint Center for Quantum Information and Computer Science, University of Maryland, USA}{xwu@cs.umd.edu}{}{}

\authorrunning{F. G. S. L. Brand\~{a}o, A. Kalev, T. Li, C. Y.-Y. Lin, K. M. Svore, X. Wu}%TODO mandatory. First: Use abbreviated first/middle names. Second (only in severe cases): Use first author plus 'et al.'

\Copyright{Fernando G. S. L. Brand\~{a}o, Amir Kalev, Tongyang Li, Cedric Y.-Y. Lin, Krysta M. Svore, Xiaodi Wu}%TODO mandatory, please use full first names. LIPIcs license is "CC-BY";  http://creativecommons.org/licenses/by/3.0/

\ccsdesc[500]{Theory of computation~Semidefinite programming}
\ccsdesc[500]{Theory of computation~Quantum query complexity}
\ccsdesc[500]{Theory of computation~Convex optimization}%TODO mandatory: Please choose ACM 2012 classifications from https://dl.acm.org/ccs/ccs_flat.cfm 

\keywords{Quantum algorithms, Semidefinite program, Convex optimization}%TODO mandatory; please add comma-separated list of keywords

\category{}%optional, e.g. invited paper

\relatedversion{A full version of the paper is available at \url{https://arxiv.org/abs/1710.02581}.}%optional, e.g. full version hosted on arXiv, HAL, or other respository/website
\supplement{}%optional, e.g. related research data, source code, ... hosted on a repository like zenodo, figshare, GitHub, ...

\nolinenumbers %uncomment to disable line numbering

\hideLIPIcs  %uncomment to remove references to LIPIcs series (logo, DOI, ...), e.g. when preparing a pre-final version to be uploaded to arXiv or another public repository

\begin{document}

\maketitle

\begin{abstract}
We give two new quantum algorithms for solving semidefinite programs (SDPs) providing quantum speed-ups. We consider SDP instances with $m$ constraint matrices, each of dimension $n$, rank at most $r$, and sparsity $s$. The first algorithm assumes an input model where one is given access to an oracle to the entries of the matrices at unit cost. We show that it has run time $\tilde{O}(s^2 (\sqrt{m} \epsilon^{-10} + \sqrt{n} \epsilon^{-12}))$, with $\epsilon$ the error of the solution. This gives an optimal dependence in terms of $m, n$ and quadratic improvement over previous quantum algorithms (when $m \approx n$). The second algorithm assumes a fully quantum input model in which the input matrices are given as quantum states. We show that its run time is $\tilde{O}(\sqrt{m}+\poly(r))\cdot\poly(\log m,\log n,B,\epsilon^{-1})$, with $B$ an upper bound on the trace-norm of all input matrices. In particular the complexity depends only polylogarithmically in $n$ and polynomially in $r$.

We apply the second SDP solver to learn a good description of a quantum state with respect to a set of measurements: Given $m$ measurements and a supply of copies of an unknown state $\rho$ with rank at most $r$, we show we can find in time $\sqrt{m}\cdot\poly(\log m,\log n,r,\epsilon^{-1})$ a description of the state as a quantum circuit preparing a density matrix which has the same expectation values as $\rho$ on the $m$ measurements, up to error $\epsilon$. The density matrix obtained is an approximation to the maximum entropy state consistent with the measurement data considered in Jaynes' principle from statistical mechanics.

As in previous work, we obtain our algorithm by "quantizing" classical SDP solvers based on the matrix multiplicative weight update method. One of our main technical contributions is a quantum Gibbs state sampler for low-rank Hamiltonians, given quantum states encoding these Hamiltonians, with a poly-logarithmic dependence on its dimension, which is based on ideas developed in quantum principal component analysis. We also develop a "fast" quantum OR lemma with a quadratic improvement in gate complexity over the construction of Harrow et al.~\cite{harrow2017sequential}. We believe both techniques might be of independent interest.
\end{abstract}

%%%%%%%%%%%%%%%%%%%%%%%%%%%%%%%%%%%%%%%%%%%%%%%%%%%%%%

\section{Introduction}
\hd{Motivation.}  Semidefinite programming has been a central topic in the study of mathematical optimization, theoretical computer science, and operations research in the last decades.  It has become an important tool for designing efficient optimization and approximation algorithms. The power of semidefinite programs (SDPs) lies in their generality (that extends the better-known linear programs (LPs)) and the fact that they admit polynomial-time solvers.

It is natural to ask whether quantum computers can have advantage in solving this important optimization problem.
In Ref.~\cite{brandao2016quantum}, Brand\~ao and Svore provided an affirmative answer, giving a quantum algorithm with worst-case running time $\tilde{O}(\sqrt{mn} s^2 (R\tilde{R}/\varepsilon)^{32})$~\footnote{$\tilde{O}$ hides factors that are polynomial in $\log m$ and $\log n$.}, where
$n$ and $s$ are the dimension and row sparsity of the input matrices, respectively,
$m$ the number of constraints, $\eps$ the accuracy of the
solution, and $R, \tilde{R}$ upper bounds on the norm of the optimal primal and dual solutions.
This is a \emph{polynomial} speed-up in $m$ and $n$ comparing to the two state-of-the-art classical SDP-solvers \cite{lee2015faster, arora2007combinatorial} (with complexity $\tilde{O}(m(m^2+n^{\omega}+mns)\poly\log(R/\eps))$ \cite{lee2015faster}, where $\omega$ is the exponent of matrix multiplication, and $\tilde{O}( m n s (R \tilde{R}/\varepsilon)^{4} + n s \left( R \tilde{R}/\varepsilon  \right)^7)$~\cite{arora2007combinatorial}), and beating the classical lower bound of $\Omega(m + n)$ \cite{brandao2016quantum}. The follow-up work by van Apeldoorn et al.~\cite{vanApeldoorn2017quantum} improved the running time giving a quantum SDP solver with complexity $\tilde{O}(\sqrt{mn} s^2 (R\tilde{R}/\eps)^8)$. In terms of limitations, Ref.~\cite{brandao2016quantum} proved a quantum lower bound $\Omega(\sqrt{m}+\sqrt{n})$ when $R,\tilde{R},s,\epsilon$ are constants; stronger lower bounds can be proven if $R$ and/or $\tilde{R}$ scale with $m$ and $n$ \cite{vanApeldoorn2017quantum}. We note all these results are shown in an input model in which there is an oracle for the entry of each of the input matrices (see \ora{plain} below for a formal definition).

In this paper, we investigate quantum algorithms for SDPs (i.e., quantum SDP solvers) further in the following two perspectives: (1) the best dependence of parameters, especially the dimension $n$ and the number of constraints $m$; (2) whether there is any reasonable alternative input model for quantum SDP solvers and what is its associated complexity. To that end, let us first formulate the precise SDP instance in our discussion.

\hd{The SDP approximate feasibility problem.} We will work with the SDP approximate feasibility problem formulated as follows  (see \append{SDP} for details): Given an $\epsilon>0$, $m$ real numbers $a_{1},\ldots,a_{m}\in\R$, and Hermitian $n\times n$ matrices $A_{1},\ldots,A_{m}$ where $-I\preceq A_{i}\preceq I,  \forall\,j\in\range{m}$, define the convex region $\mathcal{S}_{\epsilon}$ as all $X$ such that
\begin{align}
\tr(A_i X)&\leq a_{i}+\epsilon\quad\forall\,i\in\range{m};  \label{eqn:SDP}  \\
X&\succeq 0;\ \Tr[X]=1. \nonumber
\end{align}
For approximate feasibility testing, it is required that either (1) If $\mathcal{S}_{0}=\emptyset$, output fail; or (2) If $\mathcal{S}_{\epsilon}\neq\emptyset$, output an $X\in\mathcal{S}_{\epsilon}$. Throughout the paper, we denote by $n$ the the dimension of the matrices, $m$ the number of constraints, and $\epsilon$ the (additive) error of the solution. For Hermitian matrices $A$ and $B$, we denote $A\preceq B$ if $B-A$ is positive semidefinite, and $A\succeq B$ if $A-B$ is positive semidefinite. We denote $I_{n}$ to be the $n\times n$ identity matrix.

There are a few reasons that guarantee our choice of approximate SDP feasibility problem do not lose generality: (1) first, it is a routine\footnote{To see why this is the case, for any general SDP problem, one can guess a candidate value (e.g., $c_0$) for the objective function (e.g., $\tr(CX)$ and assume one wants to maximize $\tr(CX)$) and convert it into a constraint (e.g., $\tr(CX)\geq c_0$). Hence one ends up with a feasibility problem and the candidate value $c_0$ can then be found via binary search with $O(\log(1/\eps))$ overhead when $\tr(CX) \in [-1, 1]$.} to reduce general optimization SDP problems to the feasibility problem; (2) second, for general feasible solution $X\succeq 0$ with width bound $\tr(X)\leq R$, there is a procedure\footnote{The procedure goes as follows: (a) scale down every constraint by a factor $R$ and let $X'=X/R$ (thus $\tr(X')\leq 1$)
(b) let $\hat{X}=\textrm{diag}\{X, w\}$ be a block-diagonal matrix with $X$ in the upper-left corner and a scaler $w$ in the bottom-right corner. It is easy to see that $\tr(\hat{X})=1 \iff \tr(X)\leq 1$.} to derive an equivalent SDP feasibility instance with variable $\hat{X}$ s.t. $\tr(\hat{X})=1$.  Note, however, the change of $\eps$ to $\eps/R$ in this conversion. Also note one can use an approximate feasibility solver to find a strictly feasible solution, by changing $\eps$ to $\eps / R \tilde{R}$ (see Lemma 18 of Ref. \cite{brandao2016quantum}). The benefit of our choice of (\ref{eqn:SDP}) is its simplicity in presentation, which provides a better intuition behind our techniques and an easy adoption of our SDP solver in learning quantum states.
In contrast to Ref.~\cite{vanApeldoorn2017quantum}, we do not need to formulate the dual program of Eq. (\ref{eqn:SDP}) since our techniques do not rely on it.
We will elaborate more on these points in \sec{tech}.

\subsection{Quantum SDP solvers with optimal dependence on $m$ and $n$}
Existing quantum SDP solvers~\cite{brandao2016quantum,vanApeldoorn2017quantum} have close-to-optimal dependence on some key parameters but poor dependence on others. Seeking optimal parameter dependence has been an important problem in the development of classical SDP solvers and has inspired many new techniques. It is thus well motivated to investigate the optimal parameter dependence in the quantum setting.
Our first contribution is the construction of a quantum SDP solver with the optimal dependence on $m$ and $n$ in the (plain) input model as used by~\cite{brandao2016quantum, vanApeldoorn2017quantum}, given as follows:

\begin{oracle}[Plain model for $A_{j}$]\label{ora:plain}
A quantum oracle, denoted ${\cal P}_A$, such that given the indices $j \in\range{m}$, $k \in [n]$ and $l \in [s]$, computes a bit string representation of the $l$-th non-zero element of the $k$-th row of $A_j$, i.e. the oracle performs the following map:
\begin{align}
|j, k, l, z\>\rightarrow |j, k, l, z \oplus (A_{j})_{kf_{jk}(l)}\>,
\end{align}
with $f_{jk} : [r] \rightarrow [N]$ a function (parametrized by the matrix index $j$ and the row index $k$) which given $l \in [s]$ computes the column index of the $l$-th nonzero entry.
\end{oracle}

Before we move on to our main result, we will define two primitives which will appear in our quantum SDP solvers. Our main result will also be written in terms of the cost for each primitive.

\begin{definition}[trace estimation]
Assume that we have an $s$-sparse $n\times n$ Hermitian matrix $H$ with $\|H\|\leq\Gamma$ and a density matrix $\rho$. Then we define $\mathcal{S}_{\tr}(s,\Gamma,\epsilon)$ and $\mathcal{T}_{\tr}(s,\Gamma,\epsilon)$ as the number of copies of $\rho$ and the time complexity (in terms of oracle call and number of gates) of using the plain model (\ora{plain}) for $H$, respectively, such that one can compute $\tr[H\rho]$ with additive error $\epsilon$ with success probability at least $2/3$.
\end{definition}

\begin{definition}[Gibbs sampling]
Assume that we have an $s$-sparse $n\times n$ Hermitian matrix $H$ with $\|H\|\leq\Gamma$. Then we define $\mathcal{T}_{\text{Gibbs}}(s,\Gamma,\epsilon)$ as the complexity of preparing the Gibbs state $\frac{e^{-H}}{\Tr[e^{-H}]}$ with additive error $\epsilon$ using the plain model (\ora{plain}) for $H$.
\end{definition}

Our main result is as follows.

\begin{theorem}[informal; see~\thm{SDP-feasibility-testing-plain}]\label{thm:SDP-feasibility-testing-plain-intro}
In the plain input model (\ora{plain}), for any $0<\eps<1$, there is a quantum SDP solver for the feasibility problem (\ref{eqn:SDP}) using $\frac{s}{\epsilon^{4}}\tilde{O}\big(\mathcal{S}_{\tr}\big(\frac{s}{\epsilon^{2}},\frac{1}{\epsilon},\epsilon\big)\mathcal{T}_{\text{Gibbs}}\big(\frac{s}{\epsilon^{2}}, \frac{1}{\epsilon},\epsilon\big)+\sqrt{m}\mathcal{T}_{\tr}\big(\frac{s}{\epsilon^{2}},\frac{1}{\epsilon},\epsilon\big)\big)$ quantum gates and queries to \ora{plain}, where $s$ is the sparsity of $A_{j}, j \in\range{m}$.
\end{theorem}

When combined with specific instantiation of these primitives (i.e., in our case, we directly make use of results on $\mathcal{S}_{\tr}(s,\Gamma,\epsilon)$ and $\mathcal{T}_{\tr}(s,\Gamma,\epsilon)$ from Ref. ~\cite{brandao2016quantum}, and results on $\mathcal{T}_{\text{Gibbs}}(s,\Gamma,\epsilon)$ from Ref. ~\cite{poulin2009sampling}), we end up with the following concrete parameters:

\begin{corollary}[informal; see \cor{SDP-feasibility-testing-plain}] \label{cor:SDP-feasibility-testing-plain-intro}
In the plain input model (\ora{plain}), for any $0<\eps<1$, there is a quantum SDP solver for the feasibility problem (\ref{eqn:SDP}) using $\tilde{O}(s^2(\frac{\sqrt{m}}{\epsilon^{10}}+\frac{\sqrt{n}}{\epsilon^{12}}))$ quantum gates and queries to \ora{plain}, where $s$ is the sparsity of $A_{j}, j \in [m]$.
\end{corollary}

Comparing to prior art, our main contribution is to decouple the dependence on $m$ and $n$, which used to be $O(\sqrt{mn})$ and now becomes $O(\sqrt{m} +\sqrt{n})$. Note that the $(\sqrt{m}+\sqrt{n})$ dependence is optimal due to the quantum lower bound proven in Ref. \cite{brandao2016quantum}.

\begin{remark}
Even though our result achieves the optimal dependence on $m$ and $n$, it is nontrivial to obtain quantum speed-ups by directly applying our quantum SDP solvers to SDP instances from classical combinatorial problems. The major obstacle is the poly-dependence on $1/\eps$, whereas, for interesting SDP instances such as Max-Cut,  $1/\eps$ is linear in $n$. In fact, the general framework of the classical Arora-Kale SDP solver also suffers from the poly-dependence on $1/\eps$ and cannot be applied directly either. Instead, one needs to specialize the design of SDP solvers for each instance to achieve better time complexity. 

Extending this idea to quantum seems challenging. One difficulty is that known classical approaches require explicit information of intermediate states, which requires $\Omega(n)$ time and space even to store.  It is not clear how one can directly adapt classical approaches on intermediate states when stored as amplitudes in quantum states, which is the case for our current SDP solvers. It seems to us that a resolution of the problem might require an independent tool beyond the scope of this paper. We view this as an important direction for future work. 

However, our quantum SDP solvers are sufficient for instances with mild $1/\eps$, which are natural in the context of quantum information, such as learnability of the quantum state problem (elaborated in~\sec{learn_intro}) as well as examples in~\cite{vAG18}. For those cases, we do establish a quantum speed-up as any classical algorithm needs at least linear time in $n$ and/or $m$. 
\end{remark}

\subsection{Quantum SDP solvers with quantum inputs}
Given the optimality of the algorithm presented before (in terms of $m$ and $n$), a natural question is to ask about the existence of alternative input models, \emph{which can be justified for specific applications, and at the same time allows more efficient quantum SDP solvers}. This is certainly a challenging question, but we can get inspiration from the application of SDPs in quantum complexity theory (e.g., Refs.~\cite{jain2011qip,gutoski2012parallel}) and quantum information (e.g., Refs.~\cite{aaronson2007learnability, aaronson2017quantum}). In these settings, input matrices of SDP instances, with dimension $2^\ell$, are typically quantum states and/or measurements generated by $\poly(\ell)$-size circuits on $\ell$ qubits.
For the sake of these applications, it might be reasonable to equip quantum SDP solvers with the ability to leverage these circuit information, rather than merely allowing access to the entries of the input matrices.

In this paper, we propose a \emph{truly} quantum input model in which we can construct quantum SDP solvers with running time only \emph{poly-logarithmic} in the dimension.
We note that such proposal was mentioned in an earlier version of Ref.~\cite{brandao2016quantum}, whose precise mathematical form and construction of quantum SDP solvers were unfortunately incorrect, and later removed.
Note that since we consider a non-standard input model in this section, our results are incomparable to those in the plain input model. We argue for the relevance of our quantum input model, by considering an applications of the framework to the problem of learning quantum states in~\sec{learn_intro}.

\hd{Quantum input model.}
Consider a specific setting in which we are given decompositions of each $A_j$: $A_j=A_j^+ -A_j^-$, where $A_j^+, A_j^- \succeq 0$. (For instance, a natural choice is to let $A_j^+$ (resp. $A_j^-$) be the positive (resp. negative) part of $A$.)

\begin{oracle}[Oracle for traces of $A_{j}$] \label{ora:1}
A quantum oracle (unitary), denoted $O_{\Tr}$ (and its inverse $O_{\Tr}^{\dagger}$), such that for any $j\in\range{m}$,
\begin{align}
O_{\Tr}|j\>|0\>|0\>=|j\>|\Tr[A_{j}^{+}]\>|\Tr[A_{j}^{-}]\>,
\end{align}
where the real values $\Tr[A_{j}^{+}]$ and $\Tr[A_{j}^{-}]$ are encoded into their binary representations.
\end{oracle}

\begin{oracle}[Oracle for preparing $A_{j}$] \label{ora:2}
A quantum oracle (unitary), denoted $O$ (and its inverse $O^\dagger$), which acts on $\C^{m}\otimes(\C^{n}\otimes\C^{n}) \otimes (\C^{n}\otimes\C^{n})$ such that for any $j\in\range{m}$,
\begin{align}
O|j\>|0\>|0\>=|j\>|\psi_{j}^{+}\>|\psi_{j}^{-}\>,
\end{align}
where $|\psi_{j}^{+}\>, |\psi_{j}^{-}\> \in \C^{n}\otimes\C^{n}$ are any purifications of $\frac{A_{j}^{+}}{\Tr[A_{j}^{+}]}, \frac{A_{j}^{-}}{\Tr[A_{j}^{-}]}$, respectively.
\end{oracle}

\begin{oracle}[Oracle for $a_{j}$] \label{ora:3}
A quantum oracle (unitary), denoted $O_{a}$ (and its inverse $O_{a}^\dagger$), such that for any $j\in\range{m}$,
\begin{align}
O_{a}|j\>|0\>=|j\>|a_{j}\>,
\end{align}
where the real value $a_{j}$ is encoded into its binary representation.
\end{oracle}

Throughout the paper, let us assume that $A_{j}$ has rank at most $r$ for all $j\in\range{m}$ and $\Tr[A_{j}^{+}]+\Tr[A_{j}^{-}]\leq B$. The parameter $B$ is therefore an upper bound to the trace-norm of all input matrices which we assume is given as an input of the problem. Similar to the plain input model, we will define the same two primitives and their associated costs in the quantum input model.

\begin{definition}[trace estimation]\label{defn:quantum-trace-main}
We define $\mathcal{S}_{\tr}(B,\epsilon)$ and $\mathcal{T}_{\tr}(B,\epsilon)$ as the sample complexity of a state $\rho\in\C^{n\times n}$ and the gate complexity of using the quantum input oracles (\ora{1}, \ora{2}, \ora{3}), respectively, for the fastest quantum algorithm that distinguishes with success probability at least $1-O(1/m)$ whether for a fixed $j\in\range{m}$, $\tr(A_{j} \rho)>a_{j}+\epsilon$ or $\tr(A_{j} \rho)\leq a_{j}$.
 \end{definition}

\begin{definition}[Gibbs sampling]\label{defn:quantum-Gibbs-main}
Assume that $K = K^+ - K^-$, where $K^{\pm}=\sum_{j \in S} c_j A_j^{\pm}$, $c_j>0$, $S\subseteq\range{m}$ and $|S|\leq\Phi$, and that $K^+$, $K^-$ have rank at most $r_{K}$. Moreover, assume that $\tr(K^+) + \tr(K^-) \leq B_{K}$ for some $B_{K}$. Then we define $\mathcal{T}_{\text{Gibbs}}(r_{K},\Phi,B_{K},\epsilon)$ as the gate complexity of preparing the Gibbs state $\rho_{G}=\exp(-K)/\tr(\exp(-K))$ to $\epsilon$ precision in trace distance using \ora{1}, \ora{2}, and \ora{3}.
\end{definition}

Our main result in the quantum input model is as follows.

\begin{theorem}[informal; see \thm{SDP-feasibility-testing}] \label{thm:SDP-feasibility-testing-intro}
For any $\epsilon>0$, there is a quantum algorithm for the approximate feasibility of the SDP using at most $\frac{1}{\epsilon^{2}}\tilde{O}\big(\mathcal{S}_{\tr}(B,\epsilon)\mathcal{T}_{\text{Gibbs}}\big(\frac{r}{\epsilon^{2}},\frac{1}{\epsilon^{2}},\frac{B}{\epsilon},\epsilon\big)+\sqrt{m}\mathcal{T}_{\tr}(B,\epsilon)\big)$ quantum gates and queries to \ora{1}, \ora{2}, and \ora{3}.
\end{theorem}

Contrary to the plain model setting, the quantum input model is a completely new setting so that we have to construct these two primitive by ourselves. In particular, we give a construction of trace estimation in \lem{oracle-implementation-SWAP} with $\mathcal{S}_{\tr}(B,\epsilon)=\mathcal{T}_{\tr}(B,\epsilon)=O(B^{2}\log m/\epsilon^{2})$ and a construction of Gibbs sampling in \lem{gibbs_prep} with $\mathcal{T}_{\text{Gibbs}}(r_{K},\Phi,B_{K},\epsilon)=O(\Phi\cdot\poly(\log n, r_{K}, B_{K}, \epsilon^{-1}))$. As a result,

\begin{corollary}[informal; see \cor{SDP-feasibility-testing-quantum}]\label{cor:SDP-feasibility-testing-quantum-intro}
For any $\epsilon>0$, there is a quantum algorithm for the feasibility of the SDP using at most $(\sqrt{m}+\poly(r))\cdot\poly(\log m,\log n,B,\epsilon^{-1})$ quantum gates and queries to \ora{1}, \ora{2}, and \ora{3}.
\end{corollary}

We also show the square-root dependence on $m$ is also optimal by establishing the following result:
\begin{theorem}[lower bound on \cor{SDP-feasibility-testing-quantum-intro}]\label{thm:SDP-feasibility-lower-bound}
There exists an SDP feasibility testing problem such that $B,r,\epsilon=\Theta(1)$, and solving the problem requires $\Omega(\sqrt{m})$ calls to \ora{1}, \ora{2}, and \ora{3}.
\end{theorem}

\hd{Comparison between the plain model and the quantum input model.}
In the quantum input model (\ora{1}, \ora{2}, and \ora{3}), our quantum SDP solver has a \emph{poly-logarithmic} dependence on $n$ (but polynomial in $r$) and a \emph{square-root} dependence on $m$, while in the plain input model (\ora{plain}), the dependence on $n$ needs to be $\Omega(\sqrt{n})$~\cite{brandao2016quantum}. It is also worth mentioning that our quantum SDP solver in \cor{SDP-feasibility-testing-quantum-intro} does \emph{not} assume the \emph{sparsity} of $A_i$'s, which are crucial for the quantum SDP solvers with the plain model (such as \cor{SDP-feasibility-testing-plain-intro} and Refs.~\cite{brandao2016quantum,vanApeldoorn2017quantum}). This is because the quantum input models provide an alternative way to address the technical difficulty that was resolved by the sparsity condition (namely efficient algorithms for Hamiltonian evolution associated with the input matrices of the SDP).

\hd{Comparison between quantum and classical input models.}
The poly-logarithmic dependence on $n$ in \cor{SDP-feasibility-testing-quantum-intro} is intriguing and suggests that quantum computers might offer exponential speed-ups for some SDP instances. However one has to be cautious as the input model we consider is inherently quantum, so it is incomparable to classical SDP solvers. As suggested to us by Aram Harrow (personal communication), we could consider a classical setting in which we get as input all inner products between all eigenvectors of the input matrices. Then in that case one could solve the problem classically in time $\poly(r, m, 1/\epsilon)$ (essentially using Jaynes' principle which will be discussed in Section 1.5 to reduce the problem to a SDP of dimension $\poly(r)$). We have not formalized this approach, and there seems to be some technical problems doing so when the input matrices have close-by eigenvalues. However Harrow's observation shows the importance of justifying the input model in terms of natural applications to argue for the relevance of the run time obtained. We present one application of it in Section 1.5; more applications are given in Ref. \cite{vAG18}.

Furthermore, several quantum-inspired classical algorithms were recently proposed originated from Tang \cite{tang2019quantum}. Such classical algorithms assume the following sampling access:
\begin{definition}[Sampling access]\label{defn:sampling-informal}
  Let $A\in \C^{n \times n}$ be a matrix. We say that we have the \emph{sampling access} to $A$ if we can
  \begin{enumerate}
    \item\label{sample:row} sample a row index $i \in [n]$ of $A$ where the probability of row $i$ being chosen is $\frac{\|A_{i\cdot}\|^2}{\|A\|_{F}^2}$, and\footnote{Here $\|A\|_{F}$ is the Frobenius norm of $A$ and $\|A_{i\cdot}\|$ is the $\ell_{2}$ norm of the $i^{\text{th}}$ row of $A$.}
    \item\label{sample:element} for all $i \in [n]$, sample an index $j \in [n]$ where the probability of $j$ being chosen is $\frac{|A_{ij}|^2}{\|A_{i\cdot}\|^2}$
  \end{enumerate}
  with time and query complexity $O(\poly(\log n))$ for each sampling.
\end{definition}
In particular, we notice that Ref. \cite{chia2019quantum} recently gave a classical SDP solver for \eqn{SDP} with complexity $O(m\cdot\poly(\log n, r, \epsilon^{-1}))$, given the above sampling access to $A_{1},\ldots,A_{m}$. We point out that this result is incomparable to \cor{SDP-feasibility-testing-quantum-intro} because the sampling access (\defn{sampling-informal}) and our quantum state model (\ora{1}, \ora{2}, and \ora{3}) are incomparable. Nevertheless, it reminds us that under various input models, the speedup of quantum SDP solvers (compared to their classical counterparts) can also vary.

\subsection{Related works on quantum SDP solvers}
Previous quantum SDP solvers~\cite{brandao2016quantum, vanApeldoorn2017quantum} focus on the plain input model. A major contribution of ours is to improve the dependence $O(\sqrt{mn})$ to $O(\sqrt{m} +\sqrt{n})$ (ignoring dependence on other parameters) which is optimal given the lower bound $\Omega(\sqrt{m}+\sqrt{n})$ in \cite{brandao2016quantum}. To that end, we have also made a few technical contributions, including bringing in a new SDP solving framework and a fast version of quantum OR lemma (\lem{fast-quantum-OR-cite}), which will be elaborated in~\sec{tech}.

The quantum input model was briefly mentioned in an earlier version of~\cite{brandao2016quantum}. The construction of quantum SDP solvers under the quantum input model therein was unfortunately incorrect. We provide the first rigorous mathematical formulation of the quantum input model and its justification in the context of learning quantum states (see~\sec{learn_intro}). We also provide a construction of quantum SDP solvers in this model with a rigorous analysis.  Moreover, we construct the first Gibbs state sampler with quantum inputs (\lem{gibbs_prep}).

Subsequent to a previous version of this paper, an independent interesting result by van Apeldoorn and Gily{\'e}n~\cite{vAG18} has improved the complexity of trace-estimation and Gibbs sampling. After a personal communication~\cite{dewolf2017personal} introducing our fast version of the quantum OR lemma, the authors of Ref.~\cite{vAG18} observed independently that the application of the quantum OR lemma \cite{harrow2017sequential} can be applied to decouple the dependence of $m$ and $n$. As a result, Ref. \cite{vAG18} improved the complexity of \cor{SDP-feasibility-testing-plain-intro} to $\tilde{O}(s(\frac{\sqrt{m}}{\epsilon^{4}}+\frac{\sqrt{n}}{\epsilon^{5}}))$ in the quantum operator model, a stronger input model than the plain one proposed by Ref. \cite{vAG18}. Using novel techniques, it also has improved the complexity of \cor{SDP-feasibility-testing-quantum-intro} to $\tilde{O}(\frac{B\sqrt{m}}{\epsilon^{4}}+\frac{B^{3.5}}{\epsilon^{7.5}})$ in the quantum input model. Note there is no explicit dependence on the rank $r$, which is an important advance (though it can be argued that rank $r$ is implicitly included in the parameter $B$).

\subsection{Techniques}\label{sec:tech}
At a high level, and in similarity to Refs. \cite{brandao2016quantum, vanApeldoorn2017quantum}, our quantum SDP solver can be seen as a "quantized" version of classical SDP solvers based on the matrix multiplicative weight update (MMWU) method~\cite{v008a006}. In particular, we will leverage quantum Gibbs samplers as the main source of quantum speed-ups. In Refs.~\cite{brandao2016quantum, vanApeldoorn2017quantum}, quantum Gibbs samplers with quadratic speed-ups (e.g., \cite{poulin2009sampling,chowdhury2016quantum}) have been exploited to replace the classical Gibbs state calculation step in~\cite{v008a006}. Because the number of iterations in MMWU is poly-logarithmic in terms of the input size, the use of quantum Gibbs samplers, together with a few other tricks, leads to the overall quadratic quantum speed-up.

However, there are a few key differences (our major technical contributions) which are essential for our improvements.

\hd{Zero-sum game approach for MMW.} Our quantum SDP solvers do not follow the primal-dual approach in Arora-Kale's SDP solver~\cite{arora2007combinatorial} which is the classical counterpart of previous quantum SDP solvers~\cite{brandao2016quantum,vanApeldoorn2017quantum}. Instead, we follow a zero-sum game framework to solve SDP feasibility problems, which is also based on the MMWU method (details in \append{SDP}).
This framework has appeared in the classical literature (e.g.,~\cite{Hazan}) and has already been used to in semidefinite programs of relevance in quantum complexity theory (e.g.,~\cite{Wu10, gutoski2012parallel, LRS15}).
Let us briefly describe how the zero-sum game framework works when solving the SDP feasibility problem (\ref{eqn:SDP}).

Assume there are two players. Player 1 wants to provide a feasible $X \in \mathcal{S}_{\epsilon}$. Player 2, on the other side, wants to find any violation of any proposed $X$, which can be formulated as follows.
\begin{oracle}[Search for violation] \label{ora:violation-intro}
Inputs a density matrix $X$, outputs an $i\in\range{m}$ such that $\tr(A_{i} X)>a_{i}+\epsilon$. If no such $i$ exists, output "FEASIBLE".
\end{oracle}

If the original problem is feasible, there exists a feasible point $X_0$ (provided by Player 1) such that there is no violation of $X_0$ that can be found by Player 2 (i.e., \ora{violation-intro}). This actually refers to an \emph{equilibrium} point of the zero-sum game, which can also be approximated by the matrix multiplicative weight update method~\cite{v008a006}.

We argue that there are a few advantages of adopting this framework. One prominent example is its simplicity, which perhaps provides more intuition than the primal-dual approach. Together with our choice of the approximate feasibility problem, our presentation is simple both conceptually and technically (indeed, the simplicity of this framework has led to the development of the fast quantum OR lemma, another main technical contribution of ours.) Another example is that the zero-sum game approach does not make use of the dual program of SDPs and thus there is no dependence on the size of any dual solution. The game approach also admits an intuitive application of our SDP solvers to learning quantum states~\sec{learn_intro}, which coincides with the approach adopted by~\cite{LRS15} in a similar context.

One might wonder whether the simplicity of this framework will restrict the efficiency of SDP solvers. As indicated by the independent work of van Apeldoorn and Gily{\'e}n~\cite{vAG18} which has achieved the same complexity of quantum SDP solvers following both the primal-dual approach and the zero-sum approach, we conclude that it is not the case at least up to our current knowledge.

\hd{Fast quantum OR lemma.} We now outlines what is the main idea to find  a solution to \ora{violation-intro} efficiently. Roughly speaking, the idea behind previous quantum SDP solvers~\cite{brandao2016quantum, vanApeldoorn2017quantum} when applied to this context was to generate a new copy of a quantum state $X$ for each time one would query the expectation value of one of the input matrices on it. The cost of generating $X$ (i.e., Gibbs sampling) is $O(\sqrt{n})$ (ignoring the dependence on other parameters) and one can use a Grover-search-like approach to test for $m$ constraints with $O(\sqrt{m})$ iterations. The resultant cost is then $O(\sqrt{mn})$. Our key observation is to leverage the quantum OR lemma~\cite{harrow2017sequential} to detect a single violation with only a single copy of $X$.

At a high level, given a single copy of any state $\rho$ and $m$ projections $\Lambda_1, \ldots, \Lambda_m$, the quantum OR lemma describes a procedure to distinguish between the case that $\exists\,i\in\range{m}$ s.t. $\Tr[\rho\Lambda_{i}]$ is very large, or $\frac{1}{m}\sum_{i=1}^{m}\Tr[\rho\Lambda_{i}]$ is very small.
It is not hard to see that with some gap-amplification step and a search-to-decision reduction, the above procedure will output a violation $i^*$ if any.
By using quantum OR lemma, one can already decouple the cost of generating $X$ and the number of iterations in violation-detection.

Unfortunately,  Ref.~\cite{harrow2017sequential} has only been focusing on the use of a single copy of $\rho$, while its gate complexity is $O(m)$ for $m$ projections.
To optimize the gate complexity, we develop the following fast implementation of the quantum OR lemma with gate complexity $O(\sqrt{m})$, using ideas from the fast amplification technique in~\cite{nagaj2009fast}. Overall, this leads to a complexity of $O(\sqrt{m} +\sqrt{n})$.

\begin{lemma}[informal; see \lem{fast-quantum-OR-cite}]
Let $\Lambda_{1},\ldots,\Lambda_{m}$ be projections, and fix parameters $0<\varepsilon\leq 1/2$ and $\varphi, \xi>0$. Let $\rho$ be a state such that either $\exists\,j\in\range{m}$ $\Tr[\rho\Lambda_{j}]\geq 1-\varepsilon$, or $\frac{1}{m}\sum_{j=1}^{m}\Tr[\rho\Lambda_{j}]\leq\varphi$. There is a test using one copy of $\rho$ and $O(\xi^{-1}\sqrt{m}(p+\poly(\log m)))$ operations such that: in the former case, accepts with probability at least $(1-\varepsilon)^{2}/4 - \xi$; in the latter case, accepts with probability at most $3\varphi m + \xi$.
\end{lemma}

The dependence on $m$ is also tight, as one can easily embed Grover search into this problem.

\hd{Gibbs sampler with quantum inputs.} To work with the quantum input model, as our main technical contribution, we construct the first quantum Gibbs sampler of low-rank Hamiltonians when given Oracles \ref{ora:1} and \ref{ora:2}:

\begin{theorem}[informal; see \thm{low-rank-Gibbs}] \label{thm:low-rank-Gibbs-intro}
Assume the $n\times n$ matrix $K=K^+-K^-$ and $K^+, K^-$ are PSD matrices with rank at most $r_K$ and $\Tr[K^+]+\Tr[K^-]\leq B$. Given quantum oracles that prepare copies of $\rho^+=K^+ / \tr(K^+)$, $\rho^- = K^-/ \tr(K^-)$ and estimates of $\tr(K^+)$, $\tr(K^-)$, there is a quantum Gibbs sampler that prepares the Gibbs state $\rho_G=\exp(-K)/\tr(\exp(-K))$ to precision $\epsilon$ in trace distance, using $\poly(\log n, r_K, B, \epsilon^{-1})$ quantum gates.
\end{theorem}

Our quantum Gibbs sampler has a poly-logarithmic dependence on $n$ and polynomial dependence on the maximum rank of the input matrices, while in the plain input model the dependence of $n$ is $\Theta(\sqrt{n})$~ \cite{poulin2009sampling,chowdhury2016quantum}.
Our construction deviates significantly from \cite{poulin2009sampling,chowdhury2016quantum}.
Because of the existence of copies of $\rho^+$ and $\rho^-$, we rely on efficient Hamiltonian simulation techniques developed in quantum principle component analysis (PCA)~\cite{lloyd2013quantum} and its follow-up work in~\cite{kimmel2017hamiltonian}. As a result, we can also get rid of the sparsity assumption which is crucial for evoking results about efficient Hamiltonian simulation into the Gibbs sampling used in~\cite{poulin2009sampling,chowdhury2016quantum}.

\subsection{Application: Efficient learnability of quantum states}\label{sec:learn_intro}

\hd{Problem description.} Given many realizations of an experiment producing a quantum state with density matrix $\rho$, learning an approximate description of $\rho$ is a fundamental task in quantum information and experimental physics. It refers to \emph{quantum state tomography}, which has been widely used to identify quantum systems.
However, to tomograph an $\ell$-qubit state $\rho$ (with dimension $n=2^\ell$) , the optimal procedure~\cite{OW16,HHJWY16} requires $n^2$ number of copies of $\rho$, which is impractical already for relatively small $\ell$.

An interesting alternative is to find a description of the unknown quantum state $\rho$ which approximates $\Tr[\rho E_{i}]$ up to error $\eps$ for a specific collection of POVM elements $E_1,\ldots,E_m$, where $I  \succeq E_i \succeq 0$ and $E_i\in \mathbb{C}^{n\times n}, \forall i \in [m]$. This is an old problem, dating back at least to the work of Jaynes on statistical mechanics in the 50ies. Jaynes' principle \cite{jaynes1957information} (also known as the principle of maximum entropy) gives a general form for the solution of the problem above. It shows that there is always a state of the form
\begin{equation}\label{gibbsjaynes}
\frac{ \exp \left ( \sum_i \lambda_i E_i  \right) }{  \tr \left( \exp \left( \sum_i \lambda_i E_i  \right)\right)},
\end{equation}
which has the same expectation values on the $E_i$'s as the original state $\rho$, where the $\lambda_i$'s are real numbers. In words, there is always a Gibbs state with Hamiltonian given by a linear combination of the $E_i$'s which gives the same expectation values as the state described by $\rho$. Therefore one can solve the learning problem by finding the right $\lambda_i$'s (or finding a quantum circuit creating the state in Eq. (\ref{gibbsjaynes})).

\hd{Applying quantum SDP solvers.} By formulating the learning problem in terms of the SDP feasibility problem (with each $A_i$ replaced by $E_i$) where one looks for a trace unit PSD $\sigma$ matching the measurement statistics, i.e., $\Tr(\sigma E_i) \approx \Tr(\rho E_i), \forall i \in [m]$, we observe that our quantum SDP solvers actually provides a solution to the learning problem with associated speed-ups on $m$ and $n$.

In fact, our algorithm also outputs each of the $\lambda_i$'s (one can show that $\poly(\log(mn))/\varepsilon^2$ non-zero of them suffices for a solution with error $\varepsilon$), as well as a circuit description of the Gibbs state in Eq. (\ref{gibbsjaynes}) achieving the same expectation values as $\rho$ up to error $\varepsilon$.
(This is mainly because the similarity between the matrix multiplicative update method and Jaynes' principle. Compare (\ref{gibbsjaynes}) and \algo{matrixMW}.)
In this sense our result can be seen as an \emph{algorithmically} version of Jaynes' principle. We note that a similar idea was adopted by~\cite{LRS15} in learning quantum states, although for a totally different purpose (namely proving lower bounds on the size of SDP approximations to constraint satisfaction problems).

It is worthwhile noting that our quantum SDP solvers when applied in this context will output a description of the state $\rho$ in the form of Eq. (\ref{gibbsjaynes}) which has the same expectation values as $\rho$ on measurements $E_1,\ldots,E_m$ up to error $\epsilon$. This is slightly different from directly outputting estimates of $\Tr(E_i \rho)$ for each $i \in [m]$, which by itself will take $\Omega(m)$ time.

\hd{Relevance of quantum input model.} More importantly, we argue that our quantum input model is \emph{relevant} in this setting for low-rank measurements $E_i$'s.
Since all $E_i \succeq 0$ by definition, we can consider the following (slightly simplified version of) oracles:

\hd{\ora{1} for traces of $E_{i}$:} A unitary $O_{\Tr}$ such that for any $i\in\range{m}$, $O_{\Tr}|i\>|0\>=|i\>|\Tr[E_{i}]\>$.

\hd{\ora{2} for preparing $E_{i}$:} A unitary $O$ such that for any $i\in\range{m}$,
$O|i\>\<i|\otimes |0\>\<0|O^{\dagger}=|i\>\<i|\otimes |\psi_i\>\<\psi_i|$,
where $|\psi_i\>\<\psi_i|$ is any purification of $E_{i}/\Tr[E_{i}]$.

We now show how one can implement this oracle in the case where each $E_i$ is a low rank projector and we have an efficient (with $\poly \log(n)$ many gates) implementation of  the measurement. Let the rank of $E_i$'s bounded by $r$ and suppose the measurement operators $E_{i}$'s are of the form
\begin{equation}
E_i = V_{i} P_{i} V_{i}^{\cal y}
\end{equation}
for polynomial (in $\log(n)$) time circuits $V_{i}$, and projectors $P_{i}$ of the form
\begin{equation}
P_i := \sum_{i = 1}^{r_i} \ket{i}\bra{i}
\end{equation}
with $\ket{i}$ the computational basis and $r_i \leq r$. Then for \ora{1} we just need to output the $r_i$'s. \ora{2} can be implemented efficiently (in time $r \poly \log(n)$)  by first creating a maximally entangled state between the subspace spanned by $P_{i}$ and a purification and applying $V_i$ to one half of it. In more detail, consider the following purification of $E_i/ \tr(E_i)$:
\begin{equation}
\ket{\psi_i} := \frac{1}{\sqrt{r_i}}  \sum_{i=1}^{r_i} (V_i \otimes I) \ket{i, i}
\end{equation}
This can be constructed first by preparing the state $\frac{1}{\sqrt{r_i}}  \sum_{i=1}^{r_i}\ket{i, i}$ in time $r_i$ and then applying $V_i \otimes I$ to it (which can be done in time $\poly \log(n)$).

\hd{Efficient learning for low rank measurements.}  By applying our SDP solver in the quantum input model, we obtain that

\begin{theorem}[informal; see \cor{efficient-shadow}] \label{thm:learn}
For any $\epsilon>0$, there is a quantum procedure that outputs a description of the state $\rho$ in the form of Eq. (\ref{gibbsjaynes}) (namely the $\lambda_i$'s parameters) using at most $\poly(\log m,\log n,r,\epsilon^{-1})$ copies of $\rho$ and at most $\sqrt{m}\cdot\poly(\log m,\log n,r,\epsilon^{-1})$ quantum gates and queries to \ora{1} and \ora{2}.
\end{theorem}

Let us briefly sketch how our SDP solver applies to this setting. Note first that we do not aim to estimate $\Tr(E_i\rho)$ for each $i \in [m]$, which helps us circumvent the $\Omega(m)$ lower bound. What we really want is to generate a state $\tilde{\rho}$ such that $\Tr(E_i\tilde{\rho})\approx \Tr(E_i \rho)$ for each $i$. Our SDP solver will maintain and update a description of $\tilde{\rho}$ per iteration. In each iteration, given copies of $\tilde{\rho}$ and the actual unknown state $\rho$, we want to know whether $\Tr(E_i\tilde{\rho})\approx \Tr(E_i \rho)$ $\forall i\in [m]$ or there is at least a violation $i^*$. To that end, we design for each $i$ a projection for the following procedure: (1) perform multiple independent SWAP tests between $E_i/\Tr[E_i]$ (from \ora{2}) and $\rho$, $\tilde{\rho}$ respectively; (2) accept when the statistics of both SWAP tests (one with $\rho$, the other with $\tilde{\rho}$) are close. Hence, one can apply our fast quantum OR lemma on these projections to find such $i^*$ if it exists.

Note that both the sample complexity and the gate complexity of the above procedure have a poly-log dependence on $n$ (i.e., the dimension of the quantum state to learn).

\hd{Shadow tomography problem.} In a sequence of works~\cite{aaronson2007learnability, aaronson2017quantum}, Aaronson asked whether one can predict information about a dimension-$n$ quantum state with poly-log(n) many copies. In Ref.~\cite{aaronson2007learnability}, he showed that a linear number of copies is sufficient to predict the outcomes of "most" measurements according to some (arbitrary) distribution over a class of measurements. Very recently, in Ref.~\cite{aaronson2017quantum},  he referred the following problem as the "shadow tomography" problem: for any $n$-dimensional state $\rho$ and two-outcome measurements $E_{1},\ldots,E_{m}$, estimate  $\Tr[\rho E_{i}]$ up to error $\eps,$ $\forall i\in\range{m}$. He has further designed a quantum procedure for the shadow tomography problem with $\tilde{O}(\ell\cdot\log^{4}m/\eps^{5})$~\footnote{Here $\tilde{O}$ hides factors that are polynomial in $\log\log m$, $\log \log n$, and $\log 1/\epsilon$.} copies of $\rho$.

Noting that the shadow tomography problem is essentially the same problem considered by Jaynes~\cite{jaynes1957information}, one can apply Jaynes' principle and its algorithmic version we discussed before. Although this can be used to give a version of the result of Ref.~\cite{aaronson2017quantum}, Aaronson obtained his result~\cite{aaronson2017quantum} through a different route, based on a post-selection argument. A drawback of this approach is that its gate complexity is high, scaling linearly  in $m$ and as $n^{O(\log\log n)}$ (for fixed error).

Our \thm{learn} can be applied here to improve the time complexity. It gives a quantum procedure with a \emph{square-root} dependence on $m$ and $n^{O(1)}$ dependence on $n$ for arbitrary $E_i$'s.

When we assume $r$ is small, say $r=O(\poly\log n)$, the gate complexity of the entire procedure becomes $\tilde{O}(\sqrt{m}\poly \log (n))$. This gives a class of measurement (namely any set of low-rank measurements which can be efficiently implemented) for which the learning problem is efficient both in the number of samples and the computational complexity. This solves an open problem proposed in Ref.~\cite{aaronson2007learnability}

Although we have not worked out an explicit bound of the sample complexity of our procedure, the authors of~\cite{vAG18} followed our approach with more sophisticated techniques and obtained a sample complexity of $\tilde{O}(\ell\cdot\log^{4}m/\eps^{4})$, improving on the bound from~\cite{aaronson2017quantum}. We also note that very recently, Aaronson et al. claimed the same sample complexity (i.e., $\tilde{O}(\ell\cdot\log^{4}m/\eps^{4})$) in~\cite{ACHN18}.

\subsection{Open questions} This work leaves several natural open questions for future work. For example:
\begin{itemize}
\item Are there more examples of interesting SDPs where our form of input is meaningful? We have shown the example of learning quantum states. Intuitively, we are looking for SDP instances where the constraints are much "simpler" than the solution space. Is there any such example in the context of big data and/or  machine learning?
\item Our work has identified one setting where Gibbs sampling has a poly-log dependence on the dimension? Is there any other setting for the same purpose?
\item For any reasonable quantum input setting, what is the effect of potential noises on quantum inputs in practice?
\item  Can we improve further on other parameters (e.g., the dependence on $m$ and $1/\eps$)? In particular, is it possible to improve the error dependence to $\poly\log(1/\eps)$? This probably implies that we have
to consider a quantum version of the interior point method.
\item Are there other classes of measurements for which the quantum learning problem can solved in a computationally efficient way beyond the low-rank measurements we consider in this work? We note that most measurements of interest are not low rank (e.g. local measurements) and therefore the practical applicability of the present result is limited.
\end{itemize}

\hd{Organization of the appendices.}
We will formulate the SDP feasibility problem and prove the correctness of the basic framework in \append{SDP}.
Our implementation of the fast quantum OR lemma is given in~\append{fast}.
We describe our main results the constructions of quantum SDP solvers in the plain input model and the quantum input model in~\append{SDP-plain}, ~\append{SDP-efficient}, respectively.
The application to learning quantum states is illustrated in~\append{learn}.
In \append{low-rank_sampling} (with full details in \append{gibbs}) we demonstrate how to sample from the Gibbs state of low-rank Hamiltonians.

\hd{Acknowledgements.}
We thank Scott Aaronson, Joran van Apeldoorn, Andr{\'a}s Gily{\'e}n, Cupjin Huang, and anonymous reviewers for helpful discussions. We are also grateful to Joran van Apeldoorn and Andr{\'a}s Gily{\'e}n for sharing a working draft of~\cite{vAG18} with us. FB was supported by NSF. CYL and AK are supported by the Department of Defense. TL is supported by NSF CCF-1526380.  XW is supported by the U.S. Department of Energy, Office of Science, Office of Advanced Scientific Computing Research, Quantum Algorithms Teams program. XW is also supported by NSF grants CCF-1755800 and CCF-1816695.

%%%%%%%%%%%%%%%%%%%%%%%%%%%%%%%%%%%%%%%%%%%%%%%%%%%%%%%%%%%%%%%%%%%%%%%%%%%%%%

%%%%%%%%%%%%%%%%%%%%%%%%%%%%%%%%%%%%%%%%%%%%%%%%%%%%%%%%%%%%%%%%%%%%%%%%%%%%%%

\appendix

\section{Feasibility of SDPs} \label{append:SDP}
In this section, we formulate the feasibility problem of SDPs. It is a standard fact that one can use binary search to reduce any optimization problem to a feasibility one.
The high-level idea is to first guess a candidate value for the objective function, and add that as a constraint to the optimization problem. It converts the optimization problem into a feasibility problem.
One can then use binary search on the candidate value to find a good approximation to the optimal one.

\begin{definition}[Feasibility] \label{defn:feasibility}
Given an $\epsilon>0$, $m$ real numbers $a_{1},\ldots,a_{m}\in\R$, and Hermitian $n\times n$ matrices $A_{1},\ldots,A_{m}$ where $-I\preceq A_{i}\preceq I, \forall\,j\in\range{m}$, define the convex region $\mathcal{S}_{\epsilon}$ as all $X$ such that
\begin{align}
\tr(A_{i} X) &\leq a_{i}+\epsilon\quad\forall\,i\in\range{m}; \label{eqn:SDP-1} \\
X&\succeq 0; \label{eqn:SDP-2} \\
\Tr[X]&=1. \label{eqn:SDP-3}
\end{align}
For approximate feasibility testing, it is required that:
\begin{itemize}
\item If $\mathcal{S}_{0}=\emptyset$, output fail;
\item If $\mathcal{S}_{\epsilon}\neq\emptyset$, output an $X\in\mathcal{S}_{\epsilon}$.
\end{itemize}
\end{definition}

\hd{Zero-sum game approach for SDPs.} We adopt the zero-sum game approach to solve SDPs. Note that it is different from~\cite{brandao2016quantum, vanApeldoorn2017quantum} which follow the primal-dual approach of \cite{arora2007combinatorial} to solve SDPs. Instead of leveraging the dual program, we rely on the following oracle:
\begin{oracle}[Search for violation] \label{ora:violation}
Input a density matrix $X$, output an $i\in\range{m}$ such that Eq. \eqn{SDP-1} is violated. If no such $i$ exists, output "FEASIBLE".
\end{oracle}
This oracle helps establish a game view to solve any SDP feasibility problem. Imagine Player 1 who wants to provide a feasible $X \in \mathcal{S}_{\epsilon}$. Player 2, on the other side, wants to find any violation of any proposed $X$. (This is exactly the function of \ora{violation}.)
If the original problem is feasible, there exists a feasible point $X_0$ (provided by Player 1) such that there is no violation of $X_0$ that can be found by Player 2 (i.e., \ora{violation}). This actually refers to an \emph{equilibrium} point of the zero-sum game, which can be approximated by the matrix multiplicative weight update method~\cite{v008a006}.

This game view of solving the SDP feasibility problem has appeared in the classical literature (e.g.,~\cite{Hazan}) and has already been used in solving semidefinite programs in the context of quantum complexity theory (e.g.,~\cite{Wu10, gutoski2012parallel}). We observe that many techniques to quantize Arora-Kale's primal-dual approach \cite{arora2007combinatorial} for solving SDPs in Refs.~\cite{brandao2016quantum, vanApeldoorn2017quantum} readily extends to the zero-sum game approach, e.g., using quantum Gibbs samplers to generate candidate solution states.

The main difference, however, lies in the way one make use of the matrix multiplicative weight update method~\cite{kale2007efficient}, which is a meta algorithm behind both the Arora-Kale's primal-dual approach~\cite{arora2007combinatorial} and the game view approach (e.g.,~\cite{Hazan}).
As we have elaborated in~\sec{tech}, there are a few advantages of adopting this game view approach.

\hd{Master algorithm.}
We present a master algorithm that solves the SDP feasibility problem with the help of \ora{violation}. It should be understood that the master algorithm is \emph{not} the final quantum algorithm, where a few steps will be replaced by their quantum counterparts.
However, the master algorithm helps demonstrate the correctness of the algorithm and the number of oracle queries.

Our algorithm heavily relies on the matrix multiplicative weight method given in \algo{matrixMW}.

\begin{algorithm}[htbp]
\textbf{Initialization:} Fix a $\delta\leq 1/2$. Initialize the weight matrix $W^{(1)}=I_{n}$\;
\For{$t=1,2,\ldots,T$}{
	Set the density matrix $\rho^{(t)}=\frac{W^{(t)}}{\Tr[W^{(t)}]}$\;
	Observe the gain matrix $M^{(t)}$\;
	Define the new weight matrix: $W^{(t+1)}=\exp[-\delta\sum_{\tau=1}^{t}M^{(\tau)}]$\;
	}
\caption{Matrix multiplicative weights algorithm (Figure 3.1 of \cite{kale2007efficient}).}
\label{algo:matrixMW}
\end{algorithm}

\begin{proposition}[Corollary 4 of \cite{kale2007efficient}]\label{prop:matrixMW}
Assume that for all $t\in\range{T}$, either $M^{(t)}\preceq 0$ or $M^{(t)}\succeq 0$. Then \algo{matrixMW} guarantees that after $T$ rounds, for any density matrix $\rho$, we have
\begin{align}
\hspace{-2mm}(1-\delta)\sum_{t\colon M^{(t)}\preceq 0} \tr(M^{(t)} \rho^{(t)}) +(1+\delta)\sum_{t\colon M^{(t)}\succeq 0} \tr(M^{(t)}  \rho^{(t)}) \geq \sum_{t=1}^{T} \tr(M^{(t)} \rho)-\frac{\ln n}{\delta}.
\end{align}
\end{proposition}

We use \algo{matrixMW} and \prop{matrixMW} to test the feasibility of SDPs.
\begin{theorem}[Master Algorithm]\label{thm:SDP-master-1}
Assume we are given \ora{violation}. Then for any $\epsilon>0$, feasibility of the SDP in \eqn{SDP-1}, \eqn{SDP-2}, and \eqn{SDP-3} can be tested by \algo{matrixMW-gain} with at most $\frac{16\ln n}{\epsilon^{2}}$ queries to the oracle.
\end{theorem}

\begin{algorithm}[htbp]
Initialize the weight matrix $W^{(1)}=I_{n}$, and $T=\frac{16\ln n}{\epsilon^{2}}$\;
\For{$t=1,2,\ldots,T$}{
	Prepare the Gibbs state $\rho^{(t)}=\frac{W^{(t)}}{\Tr[W^{(t)}]}$\; \label{line:Gibbs}
	Find a $j^{(t)}\in\{1,2,\ldots,m\}$ such that $\tr(A_{j^{(t)}}  \rho^{(t)}) >a_{j^{(t)}}+\epsilon$ by \ora{violation}. Take $M^{(t)}=\frac{1}{2}(I_{n}-A_{j^{(t)}})$ if such $j^{(t)}$ can be found; otherwise, claim that $\mathcal{S}_{\epsilon}\neq\emptyset$, output $\rho^{(t)}$ as a feasible solution, and terminate the algorithm\; \label{line:observe-gain-matrix}
	Define the new weight matrix: $W^{(t+1)}=\exp[-\frac{\epsilon}{2}\sum_{\tau=1}^{t}M^{(\tau)}]$\;
	}
Claim that $\mathcal{S}_{0}=\emptyset$ and terminate the algorithm\;
\caption{Matrix multiplicative weights algorithm for testing the feasibility of SDPs.}
\label{algo:matrixMW-gain}
\end{algorithm}

\begin{proof}[Proof of \thm{SDP-master-1}]
For all $j\in\range{m}$, denote $M_{j}=\frac{1}{2}(I_{n}-A_{j})$; note that $0\preceq M_{j}\preceq I\ \forall j\in\range{m}$. In round $t$, after computing the density matrix $\rho^{(t)}$, equivalently speaking, \ora{violation} checks whether there exists a $j\in\range{m}$ such that $\tr(M_{j} \rho^{(t)}) <\frac{1}{2}-\frac{a_{j}+\epsilon}{2}$. If not, then $\tr(M_{j} \rho^{(t)}) \geq\frac{1}{2}-\frac{a_{j}+\epsilon}{2}\ \forall j\in\range{m}$, $\tr(A_{j}  \rho^{(t)}) \leq a_{j}+\epsilon \ \forall j\in\range{m}$, and hence $\rho^{(t)}\in\mathcal{S}_{\epsilon}$.

Otherwise, the oracle outputs an $M_{j^{(t)}}\in \{M_{j}\}_{j=1}^{m}$ such that $\tr(M_{j^{(t)}} \rho^{(t)}) <\frac{1}{2}-\frac{a_{j^{(t)}}+\epsilon}{2}$. After $T=\frac{16\ln n}{\epsilon^{2}}$ iterations, by \prop{matrixMW} (taking $\delta=\epsilon/4$ therein), this matrix multiplicative weights algorithm promises that for any density matrix $\rho$, we have
\begin{align}\label{eqn:feasibility-matrixMW}
\Big(1+\frac{\epsilon}{4}\Big)\sum_{t=1}^{T}\tr(M_{j^{(t)}} \rho^{(t)}) \geq\sum_{t=1}^{T} \tr(M_{j^{(t)}} \rho)-\frac{4\ln n}{\epsilon}.
\end{align}
If $\mathcal{S}_{0}\neq\emptyset$, there exists a $\rho^{*}\in\mathcal{S}_{0}$ such that $\tr(M_{j^{(t)}} \rho^{*}) \geq\frac{1}{2}-\frac{a_{j^{(t)}}}{2}$ for all $t\in\range{T}$. On the other hand, $\tr(M_{j^{(t)}} \rho^{(t)})<\frac{1}{2}-\frac{a_{j^{(t)}}+\epsilon}{2}$ for all $t\in\range{T}$. Plugging these two inequalities into \eqn{feasibility-matrixMW}, we have
\begin{align}
\Big(1+\frac{\epsilon}{4}\Big)\sum_{t=1}^{T}\Big(\frac{1}{2}-\frac{a_{j^{(t)}}+\epsilon}{2}\Big)> \sum_{t=1}^{T}\Big(\frac{1}{2}-\frac{a_{j^{(t)}}}{2}\Big)-\frac{4\ln n}{\epsilon},
\end{align}
which is equivalent to
\begin{align}\label{eqn:feasibility-matrixMW-1}
\frac{16\ln n}{\epsilon^{2}}>\frac{3+\epsilon}{2}T+\frac{1}{2}\sum_{t=1}^{T}a_{j^{(t)}}.
\end{align}
Furthermore, since $\frac{1}{2}-\frac{a_{j^{(t)}}}{2}\leq\tr(M_{j^{(t)}} \rho^{*})\leq 1$, we have $a_{j^{(t)}}\geq -1$\ for all $t\in\range{T}$. Plugging this into \eqn{feasibility-matrixMW-1}, we have $\frac{16\ln n}{\epsilon^{2}}>(1+\frac{\epsilon}{2})T$, and hence
\begin{align}
T<\frac{16\ln n}{\epsilon^{2}(1+\epsilon/2)}<\frac{16\ln n}{\epsilon^{2}},
\end{align}
contradiction! Therefore, if $\tr(M_{j^{(t)}} \rho^{(t)})<\frac{1}{2}-\frac{a_{j^{(t)}}+\epsilon}{2}$ happens for at least $\frac{16\ln n}{\epsilon^{2}}$ times, it must be the case that $\mathcal{S}_{0}=\emptyset$.
\end{proof}

%%%%%%%%%%%%%%%%%%%%%%%%%%%%%%%%%%%%%%%%%%%%%%%%%%%%%%%%%%%%%%%%%%%%%%%%%%%%%%

\section{Fast quantum OR lemma}\label{append:fast}
To use our master algorithm (\algo{matrixMW-gain}), a key step is to implement \ora{violation} that finds a violated constraint in the SDP. This is basically to search among $m$ measurements, which motivates us to use the quantum OR lemma from \cite{harrow2017sequential}.

\begin{lemma}[Corollary 11 of \cite{harrow2017sequential}]\label{lem:quantum-OR}
Let $\Lambda_{1},\ldots,\Lambda_{m}$ be projectors, and fix parameters $0<\epsilon\leq 1/2$, $0<\delta<1/4m$. Let $\rho$ be a state such that either $\exists\,j\in\range{m}$ such that $\Tr[\rho\Lambda_{j}]\geq 1-\epsilon$, or $\frac{1}{m}\sum_{j=1}^{m}\Tr[\rho\Lambda_{j}]\leq\delta$. Then there is a test that uses one copy of $\rho$ and: in the former case, accepts with probability at least $(1-\epsilon)^{2}/7$; in the latter case, accepts with probability at most $4\delta m$.
\end{lemma}
\noindent
However, the focus of \lem{quantum-OR} was on the single copy of $\rho$ and its proof in~\cite{harrow2017sequential} leads to a poor gate complexity.  As a result, we prove the "fast" quantum OR lemma below (\lem{fast-quantum-OR-cite}). This new version basically follows the analysis of the original quantum OR lemma; however, the projections are implemented with a quadratic speed-up in $m$ by the fast amplification technique in \cite{nagaj2009fast}. This speed-up enables us to decouple the cost of $\sqrt{m}\cdot\sqrt{n}$ in \cite{brandao2016quantum,vanApeldoorn2017quantum} to $(\sqrt{m}+\sqrt{n})$ (see \append{SDP-plain} and \append{SDP-efficient} for more details); in particular, it leads to the optimal bound for solving SDPs when other parameters are constants.

\begin{lemma}\label{lem:fast-quantum-OR-cite}
Let $\Lambda_{1},\ldots,\Lambda_{m}$ be projections, and fix parameters $0<\varepsilon\leq 1/2$ and $\varphi$. Let $\rho$ be a state such that either $\exists\,i\in\range{m}$ such that $\Tr[\rho\Lambda_{i}]\geq 1-\varepsilon$, or $\frac{1}{m}\sum_{j=1}^{m}\Tr[\rho\Lambda_{j}]\leq\varphi$. Then there is a test that uses one copy of $\rho$ and: in the former case, accepts with probability at least $(1-\varepsilon)^{2}/4 - \xi$; in the latter case, accepts with probability at most $3\varphi m + \xi$; here $\xi$ satisfies $\xi>0$ and $(1-\varepsilon)^{2}/4 - \xi>3\varphi m + \xi$. Furthermore, as long as the controlled reflection $\operatorname{ctrl}-(I-2\sum_{i = 0}^{m-1}\Lambda_{i+1}\otimes\proj{i})$ can be performed in at most $p$ operations, this test requires only $O(\xi^{-1}\sqrt{m}(p+\poly(\log m)))$ operations to complete.
\end{lemma}

\begin{proof}
Similar to \cite{harrow2017sequential}, we will reduce the task of distinguishing the two cases to estimating the eigenvalues of
\begin{equation}
\Lambda:=\frac{1}{m}\sum_{i=1}^{m}\Lambda_{i},
\end{equation}
the average of these POVM operators.
Write $P_{\geq\lambda}$ for the projector onto $\text{span}\{\ket{\lambda'}: \Lambda \ket{\lambda'} = \lambda' \ket{\lambda'},\lambda' \geq\lambda\}$. Then the following was shown in \cite{harrow2017sequential}:
\begin{lemma}[{\cite[implicit in proof of Corollary 11]{harrow2017sequential}}]
For any state $\rho$ and $\lambda \le \max_i \tr (\Lambda_i \rho)/m$,
\begin{equation}
\tr(P_{\ge \lambda}\rho) \ge [\max_i \tr (\Lambda_i \rho)-m\lambda]^2.
\end{equation}
\end{lemma}
Choose $\lambda = (1-\varepsilon)/(2m)$. Then we want to distinguish between the following two cases:
\begin{enumerate}
\item $\tr(P_{\ge \lambda } \rho) \geq (1-\varepsilon - m\lambda)^2 = (1-\varepsilon)^2/4$; \label{case:OR-case-1}
\item $\tr(\Lambda \rho) \leq\varphi$. This implies $\tr(P_{\ge 0.8\lambda}\rho) \le \varphi/(0.8\lambda) \le 3m\varphi$. \label{case:OR-case-2}
\end{enumerate}
We can explicitly decompose $\Lambda$ as follows (see also \cite[Section 2]{harrow2017sequential}): Let $Q$ be the quantum Fourier transform on $\Z_{m}$, and define the projectors $\Pi = \sum_{i = 0}^{m-1} \Lambda_{i+1} \otimes (Q\proj{i}Q^\dagger)$, $\Delta = I \otimes \proj{0}$. Then
\begin{align}
\Delta \Pi \Delta = \frac{1}{m} \sum_{i=1}^{m} \Lambda_i \otimes \proj{0} = \Lambda \otimes \proj{0}.
\end{align}
where $\proj{0}$ in the above equation is shorthand for $\proj{0}^{\ell}$ for $\ell = \lceil \log m \rceil$.

Let $a = \arccos(\sqrt{\lambda})$ and $b = \arccos (\sqrt{0.8\lambda})$. Consider the following algorithm, essentially based on the fast amplification algorithm of \cite{nagaj2009fast}:

\begin{algorithm}[H]
\begin{enumerate}
\item    Create the state $\rho \otimes \proj{0}^{\otimes \ell}$.
\item    Perform phase estimation of the rotation $(I - 2\Pi)(I - 2\Delta)$ on the state, with precision $(b-a)/2$ and error probability $\xi$. Let the measured eigenvalue be $\phi$.
\item   Accept iff $|\phi| \le (a+b)/2$.
\end{enumerate}
\caption{The fast amplification algorithm in \cite{nagaj2009fast}.}
\label{algo:fast-amplification}
\end{algorithm}

The following lemma in \cite{nagaj2009fast} follows from a direct application of Jordan's lemma:
\begin{lemma}\cite[Section 2.1]{nagaj2009fast}
If $\ket{\psi} \otimes \ket{0}^{\otimes \ell}$ is an eigenvector of $\Delta\Pi\Delta$ with eigenvalue $\cos^2\phi$, then
\begin{equation}
\ket{\psi} \otimes \ket{0}^{\otimes \ell} = \frac{1}{\sqrt{2}} (\ket{\phi} + \ket{-\phi})
\end{equation}
where $\ket{\phi}$ and $\ket{-\phi}$ are some eigenvectors of $(I - 2\Pi)(I - 2\Delta)$ with eigenvalues $\phi$ and $-\phi$, respectively.
 \end{lemma}

In Case 1, we have $\tr(P_{\ge \lambda}\rho) \ge  (1-\varepsilon)^2/4$, and therefore Algorithm \ref{algo:fast-amplification} accepts with probability at least  $(1-\varepsilon)^2/4-\xi$. In Case 2, we have $\tr(P_{\ge 0.8\lambda}\rho) \le 3m\varphi$, and therefore Algorithm \ref{algo:fast-amplification} accepts with probability at most $3m\varphi + \xi$.

Algorithm \ref{algo:fast-amplification} requires applying the controlled version of the Grover iterate $(I - 2\Pi)(I - 2\Delta)$ $O(((b-a)\xi)^{-1}) = O(\sqrt{m}\xi^{-1})$ times. Furthermore, the controlled reflection ctrl-$(I - 2\Delta)$ is implementable by $O(\log m)$ gates since $\Delta = I \otimes \proj{0}^{\otimes \lceil \log m \rceil}$, and the controlled reflection ctrl-$(I - 2\Pi)$ is implementable using $O(p+\poly(\log m))$ gates by assumption.
\end{proof}

\begin{remark}
The gate complexity in \lem{fast-quantum-OR-cite} is optimal in $\sqrt{m}$, i.e., there exists projections $\Lambda_{1},\ldots,\Lambda_{m}$ and a state $\rho$ such that distinguishing whether $\exists\,i\in\range{m}$ $\Tr[\rho\Lambda_{i}]\geq 2/3$ or $\frac{1}{m}\sum_{j=1}^{m}\Tr[\rho\Lambda_{j}]\leq 1/8m$ requires at least $\Omega(\sqrt{m})$ gates. In particular, assume that $\Lambda_{i}=|i\>\<i|$ for all $i\in\range{m}$ and $\rho=|k\>\<k|$ where $k\in\range{m+1}$. Then to distinguish whether $\exists\,i\in\range{m}$ $\Tr[\rho\Lambda_{i}]\geq 2/3$ or $\frac{1}{m}\sum_{j=1}^{m}\Tr[\rho\Lambda_{j}]\leq 1/8m$, it is equivalent to searching whether $k\in\range{m}$ or not; deciding this requires at least $\Omega(\sqrt{m})$ gates due to the hardness of Grover search \cite{bennett1997strengths}.
\end{remark}

%%%%%%%%%%%%%%%%%%%%%%%%%%%%%%%%%%%%%%%%%%%%%%%%%%%%%%%%%%%%%%%%%%%%%%%%%%%%%%%

\section{Quantum SDP solver in the plain model}\label{append:SDP-plain}
Before we get into the quantum SDP solver in the plain model, we first modularize the cost of two important blocks as follows.

\begin{definition}[trace estimation]\label{defn:plain-trace}
Assume that we have an $s$-sparse $n\times n$ Hermitian matrix $H$ with $\|H\|\leq\Gamma$ and a density matrix $\rho$. Then we define $\mathcal{S}_{\tr}(s,\Gamma,\epsilon)$ and $\mathcal{T}_{\tr}(s,\Gamma,\epsilon)$ as the sample complexity of $\rho$ and the time complexity of using the plain model (\ora{plain}) of $H$ and two-qubit gates, respectively, such that one can compute $\tr[H\rho]$ with additive error $\epsilon$ with success probability at least $2/3$.
\end{definition}

\begin{definition}[Gibbs sampling]\label{defn:plain-Gibbs}
Assume that we have an $s$-sparse $n\times n$ Hermitian matrix $H$ with $\|H\|\leq\Gamma$. Then we define $\mathcal{T}_{\text{Gibbs}}(s,\Gamma,\epsilon)$ as the complexity of preparing the Gibbs state $\frac{e^{-H}}{\Tr[e^{-H}]}$ with additive error $\epsilon$ using the plain model (\ora{plain}) of $H$ and two-qubit gates.
\end{definition}

As a subsequence of \lem{fast-quantum-OR-cite}, \defn{plain-trace}, and \defn{plain-Gibbs}, we prove the following theorem under the plain model:

\begin{theorem}\label{thm:SDP-feasibility-testing-plain}
Assume we are given \ora{plain}. Furthermore, assume that $A_{j}$ is $s$-sparse for all $j\in\range{m}$. Then for any $\epsilon>0$, feasibility of the SDP in \eqn{SDP-1}, \eqn{SDP-2}, and \eqn{SDP-3} can be tested by \algo{efficient-SDP-plain} with success probability at least 0.96 and $\frac{s}{\epsilon^{4}}\tilde{O}\big(\mathcal{S}_{\tr}\big(\frac{s}{\epsilon^{2}},\frac{1}{\epsilon},\epsilon\big)\mathcal{T}_{\text{Gibbs}}\big(\frac{s}{\epsilon^{2}},\frac{1}{\epsilon},\epsilon\big)+\sqrt{m}\mathcal{T}_{\tr}\big(\frac{s}{\epsilon^{2}},\frac{1}{\epsilon},\epsilon\big)\big)$ quantum gates and queries to \ora{plain}.
\end{theorem}

\begin{algorithm}[htbp]
Initialize the weight matrix $W^{(1)}=I_{n}$, and $T=\frac{16\ln n}{\epsilon^{2}}$\;
\For{$t=1,2,\ldots,T$}{
	Prepare $\log m\cdot\mathcal{S}_{\tr}(\frac{s\log n}{\epsilon^{2}},\frac{\log n}{\epsilon},\epsilon)$ samples of Gibbs state $\rho^{(t)}=\frac{W^{(t)}}{\Tr[W^{(t)}]}$ by \defn{plain-Gibbs}\; \label{line:Gibbs-efficient-plain}
    Using these $\log m\cdot\mathcal{S}_{\tr}(\frac{s\log n}{\epsilon^{2}},\frac{\log n}{\epsilon},\epsilon)$ copies of $\rho^{(t)}$, search for a $j^{(t)}\in\range{m}$ such that $\Tr[A_{j^{(t)}} \rho^{(t)}]>a_{j^{(t)}}+\epsilon$ by \lem{fast-quantum-OR-cite} (for each $j$, we use \defn{plain-trace} to compute $\Tr[A_{j}\rho]$). If such $j^{(t)}$ is found, take $M^{(t)}=\frac{1}{2}(I_{n}-A_{j^{(t)}})$; otherwise, claim that $\mathcal{S}_{\epsilon}\neq\emptyset$ (the SDP is feasible)\; \label{line:observe-gain-matrix-efficient-plain}
	Define the new weight matrix: $W^{(t+1)}=\exp[-\frac{\epsilon}{4}\sum_{\tau=1}^{t}M^{(\tau)}]$\;
	}
Claim that $\mathcal{S}_{0}=\emptyset$ and terminate the algorithm.
\caption{Efficiently testing the feasibility of SDPs: Plain model.}
\label{algo:efficient-SDP-plain}
\end{algorithm}

\begin{proof}[Proof of \thm{SDP-feasibility-testing-plain}]
The correctness of \algo{efficient-SDP-plain} is automatically established by \thm{SDP-master-1}; it suffices to analyze the gate cost of \algo{efficient-SDP-plain}.

In Line \ref{line:Gibbs-efficient-plain} of \algo{efficient-SDP-plain}, we apply \defn{plain-Gibbs} to compute the Gibbs state $\rho^{(t)}$. In round $t$, because $t\leq\frac{16\ln n}{\epsilon^{2}}$, $\frac{\epsilon}{4}\sum_{\tau=1}^{t}M^{(\tau)}$ has sparsity at most $s'\leq t\cdot s=O(\frac{s\log n}{\epsilon^{2}})$, and $\|\frac{\epsilon}{4}\sum_{\tau=1}^{t}M^{(\tau)}\|\leq \frac{\epsilon}{4}\cdot t=O(\frac{\log n}{\epsilon})$. As a result, $\mathcal{T}_{\text{Gibbs}}(\frac{s\log n}{\epsilon^{2}},\frac{\log n}{\epsilon},\epsilon)$ quantum gates and queries to \ora{plain} suffice to prepare a copy of the Gibbs state $\rho^{(t)}$. In addition, since to query an element of $\frac{\epsilon}{4}\sum_{\tau=1}^{t}M^{(\tau)}$ we need to query each of the $A_{j^{(\tau)}}$, we have an overhead of $s\cdot\frac{16\ln n}{\epsilon^{2}}$ for constructing \ora{plain} for $\frac{\epsilon}{4}\sum_{\tau=1}^{t}M^{(\tau)}$ (in particular, Appendix D of the full version of \cite{vanApeldoorn2017quantum} showed that this overhead $\Theta(\frac{s\ln n}{\epsilon^{2}})$ is necessary and sufficient for constructing the plain oracle for $\frac{\epsilon}{4}\sum_{\tau=1}^{t}M^{(\tau)}$). In total, Line \ref{line:Gibbs-efficient-plain} of \algo{efficient-SDP-plain} costs
\begin{align}
\frac{16s\ln n}{\epsilon^{2}}\cdot\log m\cdot\mathcal{S}_{\tr}\Big(\frac{s\log n}{\epsilon^{2}},\frac{\log n}{\epsilon},\epsilon\Big)\cdot \mathcal{T}_{\text{Gibbs}}\Big(\frac{s\log n}{\epsilon^{2}},\frac{\log n}{\epsilon},\epsilon\Big)
\end{align}
quantum gates and queries to \ora{plain}.

Next, using these $\log m\cdot\mathcal{S}_{\tr}(\frac{s\log n}{\epsilon^{2}},\frac{\log n}{\epsilon},\epsilon)$ copies of $\rho^{(t)}$, we apply \defn{plain-trace} for $O(\log m)$ times to create two-outcome POVMs $M_j$ for any $j\in\range{m}$ such that $M_{j}$ decides whether $\tr(A_{j}\rho)-a_{j}>\epsilon$ with success probability boosted to $1-O(1/m)$. The gate complexity of each $M_{j}$ is $\mathcal{T}_{\tr}(\frac{s\log n}{\epsilon^{2}},\frac{\log n}{\epsilon},\epsilon)$ by \defn{plain-trace}. Furthermore, because \ora{plain} is reversible, we can assume an explicit decomposition $M_{j} \otimes \proj{0}^{\otimes a} = P \Lambda_{j} P$, for some integer $a$, $P = I \otimes \proj{0}^{\otimes a}$, and some orthogonal projector $\Lambda_{j}$. Let $\widetilde{\rho} = \rho^{\otimes C}\otimes\proj{0}^a$ where $C=\log m\cdot\mathcal{S}_{\tr}(\frac{s\log n}{\epsilon^{2}},\frac{\log n}{\epsilon},\epsilon)$ for a large enough constant in the big-$O$. We therefore need to decide between the cases
\begin{enumerate}
\item $\Tr[\Lambda_j\widetilde{\rho}]\geq 1-\frac{0.01}{m}$ for some $j\in\range{m}$; or
\item $\Tr[\Lambda_j\widetilde{\rho}]\leq \frac{0.01}{m}$ for all $j\in\range{m}$.
\end{enumerate}
This corresponds to the two cases of \lem{fast-quantum-OR-cite}, where $\varepsilon=\varphi=\frac{0.01}{m}$. Because each $M_{j}$ can be implemented with $\mathcal{T}_{\tr}(\frac{s\log n}{\epsilon^{2}},\frac{\log n}{\epsilon},\epsilon)$ two-qubit gates, the total gate complexity of implementing the reflection $I-2\sum_{j=0}^{m-1}\Lambda_{j}\otimes|j\>\<j|$ is also $\mathcal{T}_{\tr}(\frac{s\log n}{\epsilon^{2}},\frac{\log n}{\epsilon},\epsilon)$. As a result, the total cost of applying \lem{fast-quantum-OR-cite} is $\tilde{O}\big(\sqrt{m}\mathcal{T}_{\tr}(\frac{s\log n}{\epsilon^{2}},\frac{\log n}{\epsilon},\epsilon)\big)$.

In \lem{fast-quantum-OR-cite}, we choose $\xi=\frac{1}{3}(\frac{(1-\varepsilon)^{2}}{4}-3m\varphi)$ -- this is a positive constant. We can thus tell the two cases apart with constant probability. Then, we repeat the call of \lem{fast-quantum-OR-cite} for $L=\Theta(\log\frac{\log n}{\epsilon^{2}})$ times and accept if and only if \lem{fast-quantum-OR-cite} accepts for at least $\frac{L}{2}\cdot(\frac{(1-\epsilon)^{2}}{4}+3m\varphi)$ times. By Chernoff's bound, this can enhance the success probability to at least $1-\frac{\epsilon^{2}}{400\ln n}$.

In all, we have a quantum algorithm that determines whether there exists a $j\in\range{m}$ such that $\tr[A_{j} \rho]\geq a_{j}+\epsilon$ with success probability at least $1-\frac{\epsilon^{2}}{400\ln n}$, using $\tilde{O}\big(\sqrt{m}\mathcal{T}_{\tr}(\frac{s\log n}{\epsilon^{2}},\frac{\log n}{\epsilon},\epsilon)\big)$ quantum gates and queries to \ora{plain}. To find this $j$, we apply binary search on $j\in\{1,2,\ldots,m\}$, i.e., apply the algorithm to $j\in\{1,\ldots,\lfloor m/2\rfloor\}$ and $j\in\{\lceil m/2\rceil,\ldots,m\}$ respectively, and if the output is yes then call the algorithm recursively. This gives an extra $\poly(\log m)$ overhead on the queries to \ora{plain}, which is still $\tilde{O}\big(\sqrt{m}\mathcal{T}_{\tr}(\frac{s\log n}{\epsilon^{2}},\frac{\log n}{\epsilon},\epsilon)\big)$. In addition, similar to the analysis of Line \ref{line:Gibbs-efficient-plain}, there is an overhead of $s\cdot\frac{16\ln n}{\epsilon^{2}}$ for constructing \ora{plain} of the Gibbs state using \ora{plain} of each of the $A_{j^{(\tau)}}$. Therefore, the total cost of executing Line \ref{line:observe-gain-matrix-efficient-plain} of \algo{efficient-SDP-plain} is
\begin{align}
\frac{s}{\epsilon^{2}}\tilde{O}\Big(\sqrt{m}\mathcal{T}_{\tr}\Big(\frac{s\log n}{\epsilon^{2}},\frac{\log n}{\epsilon},\epsilon\Big)\Big).
\end{align}

Because \algo{efficient-SDP-plain} has at most $\frac{16\ln n}{\epsilon^{2}}$ iterations, with success probability at least $1-\frac{16\ln n}{\epsilon^{2}}\cdot\frac{\epsilon^{2}}{400\ln n}=0.96$ \algo{efficient-SDP-plain} works correctly, and its execution takes
\begin{align}
&\frac{16s\ln n}{\epsilon^{4}}\cdot\Big(\log m\cdot\mathcal{S}_{\tr}\Big(\frac{s\log n}{\epsilon^{2}},\frac{\log n}{\epsilon},\epsilon\Big)\cdot \mathcal{T}_{\text{Gibbs}}\Big(\frac{s\log n}{\epsilon^{2}},\frac{\log n}{\epsilon},\epsilon\Big)+\tilde{O}\Big(\sqrt{m}\mathcal{T}_{\tr}\big(\frac{s\log n}{\epsilon^{2}},\frac{\log n}{\epsilon},\epsilon\big)\Big)\Big) \nonumber \\
&\quad=\frac{s}{\epsilon^{4}}\tilde{O}\Big(\mathcal{S}_{\tr}\Big(\frac{s}{\epsilon^{2}},\frac{1}{\epsilon},\epsilon\Big)\mathcal{T}_{\text{Gibbs}}\Big(\frac{s}{\epsilon^{2}},\frac{1}{\epsilon},\epsilon\Big)+\sqrt{m}\mathcal{T}_{\tr}\Big(\frac{s}{\epsilon^{2}},\frac{1}{\epsilon},\epsilon\Big)\Big).
\end{align}
two-qubit gates and queries to \ora{plain}.
\end{proof}

To be more explicit, the complexities of $\mathcal{S}_{\tr}$, $\mathcal{T}_{\tr}$, and $\mathcal{T}_{\text{Gibbs}}$ are given in previous literatures:

\begin{lemma}[Lemma 12, \cite{brandao2016quantum}]\label{lem:BS-trace}
Given an $s$-sparse $n\times n$ Hermitian matrix $H$ with $\|H\|\leq 1$ and a density matrix $\rho$, with probability larger than $1-p_{e}$, one can compute $\tr[H\rho]$ with additive error $\epsilon$ in time $O(s\epsilon^{-2}\log^{4}(ns/p_{e}\epsilon))$ using $O(\epsilon^{-2}\log(1/p_{e}))$ copies of $\rho$. In other words, $\mathcal{S}_{\tr}(s,1,\epsilon)=O(1/\epsilon^{2})$ and $\mathcal{T}_{\tr}(s,1,\epsilon)=O(s/\epsilon^{2})$.
\end{lemma}

\begin{lemma}[\cite{poulin2009sampling}]\label{lem:PW-Gibbs}
Given an $s'$-sparse $n\times n$ Hermitian matrix $H$ with $\|H\|\leq \beta$ for some $\beta>0$, one can prepare the Gibbs state $\frac{e^{-H}}{\Tr[e^{-H}]}$ with additive error $\epsilon$ using $\tilde{O}(\frac{\sqrt{\dim(H)}\beta s'}{\epsilon})$ calls to \ora{plain} of $H$ and two-qubit gates. In other words, $\mathcal{T}_{\text{Gibbs}}(s,\Gamma,\epsilon)=\tilde{O}(s\Gamma\sqrt{n}/\epsilon)$.
\end{lemma}

As a consequence of \thm{SDP-feasibility-testing-plain}, \lem{BS-trace}, and \lem{PW-Gibbs}, we have the following complexity result for solving SDPs under the plain model:

\begin{corollary}\label{cor:SDP-feasibility-testing-plain}
Assume we are given \ora{plain}. Furthermore, assume that $A_{j}$ is $s$-sparse for all $j\in\range{m}$. Then for any $\epsilon >0$, feasibility of the SDP in \eqn{SDP-1}, \eqn{SDP-2}, and \eqn{SDP-3} can be tested by \algo{efficient-SDP-plain} with success probability at least 0.96 and $\tilde{O}(s^{2}(\frac{\sqrt{m}}{\epsilon^{10}}+\frac{\sqrt{n}}{\epsilon^{12}}))$ quantum gates and queries to \ora{plain}.
\end{corollary}

\begin{proof}
Note that $\mathcal{S}_{\tr}(s,\Gamma,\epsilon)=\mathcal{S}_{\tr}(s,1,\frac{\epsilon}{\Gamma})$ and $\mathcal{T}_{\tr}(s,\Gamma,\epsilon)=\mathcal{T}_{\tr}(s,1,\frac{\epsilon}{\Gamma})$ by renormalizing the Hamiltonian $H$ to $H/\Gamma$. As a result, plugging \lem{BS-trace} and \lem{PW-Gibbs} into \thm{SDP-feasibility-testing-plain}, the complexity of solving the SDP becomes
\begin{align}
\frac{s}{\epsilon^{4}}\cdot\tilde{O}\Big(\frac{1}{\epsilon^{4}}\cdot\frac{s\sqrt{n}}{\epsilon^{4}}+\frac{s\sqrt{m}}{\epsilon^{6}}\Big)=\tilde{O}\Big(s^{2}\Big(\frac{\sqrt{m}}{\epsilon^{10}}+\frac{\sqrt{n}}{\epsilon^{12}}\Big)\Big).
\end{align}
\end{proof}

\begin{remark}
The $(\sqrt{m}+\sqrt{n})$ dependence is optimal compared to \cite{brandao2016quantum,vanApeldoorn2017quantum}.
\end{remark}

\begin{remark}
Using more elaborated techniques and analyses, Ref. \cite{vAG18} improved the complexity of \cor{SDP-feasibility-testing-plain} to $\tilde{O}(s(\frac{\sqrt{m}}{\epsilon^{4}}+\frac{\sqrt{n}}{\epsilon^{5}}))$.
\end{remark}

%%%%%%%%%%%%%%%%%%%%%%%%%%%%%%%%%%%%%%%%%%%%%%%%%%%%%%%%%%%%%%%%%%%%%%%%%%%%%%%

\section{Quantum SDP solver with quantum inputs} \label{append:SDP-efficient}
In this section, we illustrate our quantum SDP solver in the quantum input model. To that end, we first provide a precise formulation of the quantum input model, and then demonstrate how to implement \ora{violation} in such scenario and how the actual quantum algorithm works.

%=============================================================================
\subsection{The quantum input model}
As mentioned in the introduction, we would like to equip the quantum SDP solver with some extra power beyond only accessing the entries of the input matrices (i.e., $A_j$, $j=1,\ldots,m$, each of $n\times n$ size). We imagine the setting where these $A_j$s are nice so that the following oracles, representing various means to access $A_j$s, can be efficiently implemented.

\begin{oracle}[Oracle for traces of $A_{j}$] \label{ora:trace}
A quantum oracle (unitary), denoted $O_{\Tr}$ (and its inverse $O_{\Tr}^{\dagger}$), such that for any $j\in\range{m}$,
\begin{align}
O_{\Tr}|j\>|0\>|0\>=|j\>|\Tr[A_{j}^{+}]\>|\Tr[A_{j}^{-}]\>,
\end{align}
where $A_{j}^{+}$ and $A_{j}^{-}$ are two PSD matrices such that $A_{j}=A_{j}^{+}-A_{j}^{-}$ (the real values $\Tr[A_{j}^{+}]$ and $\Tr[A_{j}^{-}]$ are encoded into their binary representations).
\end{oracle}

\begin{oracle}[Oracle for preparing $A_{j}$] \label{ora:prep}
A quantum oracle (unitary), denoted $O$ (and its inverse $O^\dagger$), which acts on $\C^{m}\otimes(\C^{n}\otimes\C^{n}) \otimes (\C^{n}\otimes\C^{n})$ such that for any $j\in\range{m}$,
\begin{align}
O|j\>\<j|\otimes |0\>\<0|\otimes |0\>\<0|O^{\dagger}=|j\>\<j|\otimes |\psi_{j}^{+}\>\<\psi_{j}^{+}| \otimes |\psi_{j}^{-}\>\<\psi_{j}^{-}|,
\end{align}
where $|\psi_{j}^{+}\>, |\psi_{j}^{-}\> \in \C^{n}\otimes\C^{n}$ are any purifications of $\frac{A_{j}^{+}}{\Tr[A_{j}^{+}]}, \frac{A_{j}^{-}}{\Tr[A_{j}^{-}]}$, respectively.\footnote{By tracing out the extra space, one can easily obtain states $A_{j}^{+}/\Tr[A_{j}^{+}], A_{j}^{-}/\Tr[A_{j}^{-}]$.}
\end{oracle}

\begin{oracle}[Oracle for $a_{j}$]\label{ora:a}
A quantum oracle (unitary), denoted $O_{a}$ (and its inverse $O_{a}^\dagger$), such that for any $j\in\range{m}$,
\begin{align}
O_{a}|j\>\<j|\otimes |0\>\<0|O_{a}^{\dagger}=|j\>\<j|\otimes |a_{j}\>\<a_{j}|,
\end{align}
where the real value $a_{j}$ is encoded into its binary representation.
\end{oracle}

Similar to \append{SDP-plain}, we also modularize the cost of two important blocks as follows.

\begin{definition}[trace estimation]\label{defn:quantum-trace}
Assume that $\Tr(A_{j}^{+})+\Tr(A_{j}^{-})\leq B$ for some bound $B$ for all $j\in\range{m}$. Then we define $\mathcal{S}_{\tr}(B,\epsilon)$ and $\mathcal{T}_{\tr}(B,\epsilon)$ as the sample complexity of a state $\rho\in\C^{n\times n}$ and the gate complexity of using the quantum input oracles (\ora{trace}, \ora{prep}, \ora{a}), and two-qubit gates, respectively, such that there exists a quantum algorithm which distinguishes with success probability at least $1-O(1/m)$ whether for a fixed $j\in\range{m}$, $\tr(A_{j} \rho)>a_{j}+\epsilon$ or $\tr(A_{j} \rho)\leq a_{j}$.
\end{definition}

\begin{definition}[Gibbs sampling]\label{defn:quantum-Gibbs}
Assume that $K = K^+ - K^-$, where $K^{\pm}=\sum_{j \in S} c_j A_j^{\pm}$,  $S\subseteq\range{m}$ and $|S|\leq\Phi$, $c_j>0$, and $A_j^{\pm}$ refers to either $A_j^{+}$ or $A_j^{-}$ for all $j\in\range{m}$. Moreover, assume that $\tr(K^+) + \tr(K^-) \leq B_{K}$ for some bound $B_{K}$, and that $K^+$, $K^-$ have rank at most $r_{K}$. Then we define $\mathcal{T}_{\text{Gibbs}}(r_{K},\Phi,B_{K},\epsilon)$ as the complexity of preparing the Gibbs state $\rho_{G}=\exp(-K)/\tr(\exp(-K))$ to $\epsilon$ precision in trace distance using \ora{trace}, \ora{prep}, \ora{a}, and two-qubit gates.
\end{definition}

%=============================================================================
\subsection{Implementation of \ora{violation} -- searching a violated constraint}
Using \ora{trace}, \ora{prep}, and \ora{a}, \ora{violation} can be implemented by the following lemma, using our fast quantum OR lemma (\lem{fast-quantum-OR-cite}):
\begin{lemma}\label{lem:oracle-implementation}
Given $\epsilon,\delta\in (0,1)$. Assume we have \ora{trace}, \ora{prep}, \ora{a}, and $(\log 1/\delta)\cdot\tilde{O}(\mathcal{S}_{\tr}(B,\epsilon))$ copies of a state $\rho$. Assume either $\exists\,j\in\range{m}$ such that $\tr(A_{j} \rho)\geq a_{j}+\epsilon$, or $\tr(A_{j} \rho)\leq a_{j}$ for all $j\in\range{m}$. Then there is an algorithm that in the former case, finds such a $j$; and in the latter case, returns "FEASIBLE". This algorithm has success probability $1-\delta$ and uses in total $\log 1/\delta\cdot\tilde{O}(\sqrt{m}\mathcal{T}_{\tr}(B,\epsilon))$ quantum gates and queries to \ora{trace}, \ora{prep}, and \ora{a}.
\end{lemma}

\begin{proof}
First, we use \defn{quantum-trace} to create two-outcome POVMs $M_j$, acting on $\rho$, $|\psi_{j}^{+}\>\<\psi_{j}^{+}|$, and $|\psi_{j}^{-}\>\<\psi_{j}^{-}|$ with $C=\mathcal{S}_{\tr}(B,\epsilon)$ copies, such that $M_{j}$ decides with probability $1-O(1/\poly(m))$ whether $\tr(A_{j}\rho)-a_{j}>\epsilon$.

Because we are given purifications of all $A_{j}^{+}$ and $A_{j}^{-}$ in \ora{prep}, for all $j\in\{1,\ldots,m\}$ we can assume an explicit decomposition $M_{j} \otimes \proj{0}^{\otimes a} = P \Lambda_{j} P$, for some integer $a$, $P = I \otimes \proj{0}^{\otimes a}$, and some orthogonal projector $\Lambda_{j}$. Let $\widetilde{\rho} = \rho^{\otimes C} \otimes (|\psi_{j}^{+}\>\<\psi_{j}^{+}|)^{\otimes C} \otimes (|\psi_{j}^{-}\>\<\psi_{j}^{-}|)^{\otimes C} \otimes \proj{0}^a$. We therefore need to decide between the cases
\begin{enumerate}
\item $\Tr[\Lambda_j\widetilde{\rho}] \ge 1-O(1/\poly(m)) $ for some $j$; or
\item $\Tr[\Lambda_j\widetilde{\rho}] \le O(1/\poly(m)) $ for all $j$.
\end{enumerate}
This corresponds to the two cases of \lem{fast-quantum-OR-cite}, where both $\varepsilon$ and $\delta$ are $O(1/\poly(m))$. To implement the the projection $I-2\sum_{j = 1}^{m}\Lambda_{j}\otimes\proj{j}$ in \lem{fast-quantum-OR-cite}, we use \ora{prep} to obtain purifications $|\psi_{j}^{+}\>\<\psi_{j}^{+}|$ and $|\psi_{j}^{-}\>\<\psi_{j}^{-}|$ of $\frac{A_{j}^{+}}{\Tr[A_{j}^{+}]}$ and $\frac{A_{j}^{-}}{\Tr[A_{j}^{-}]}$, and apply the reflection with respect to $|\psi_{j}^{+}\>$ and $|\psi_{j}^{-}\>$; note that we can obtain the numbers $\Tr[A_{j}^{+}]$ and $\Tr[A_{j}^{-}]$ in superposition by \ora{trace}. Including the controlling ancilla $|j\>\<j|$, the $p$ in \lem{fast-quantum-OR-cite} is at most $O(\log m)$.

In \lem{fast-quantum-OR-cite}, choose $\xi=\frac{1}{3}(\frac{(1-\varepsilon)^{2}}{4}-3m\varphi)$ -- this is a positive constant. We can thus tell the two cases apart with constant probability, using $\mathcal{S}_{\tr}(B,\epsilon)$ samples of $\rho$ and $\tilde{O}(\sqrt{m})\cdot\mathcal{T}_{\tr}(B,\epsilon)$ other operations. Then, we repeat the call of \lem{fast-quantum-OR-cite} for $L=\Theta(\log\delta^{-1})$ times and accept if and only if \lem{fast-quantum-OR-cite} accepts for at least $\frac{L}{2}\cdot(\frac{(1-\epsilon)^{2}}{4}+3m\varphi)$ times. By Chernoff's bound, this enhances the success probability to at least $1-\delta$.

In all, we have a quantum algorithm that determines whether there exists a $j\in\range{m}$ such that $\tr(A_{j} \rho)\geq a_{j}+\epsilon$ (or $\tr(A_{j} \rho)\leq a_{j}$ for all $j\in\range{m}$) with success probability at least $1-\delta$, using $\log 1/\delta\cdot\tilde{O}(\sqrt{m}\mathcal{T}_{\tr}(B,\epsilon))$ quantum gates and queries to \ora{trace}, \ora{prep}, and \ora{a}. To find this $j$, we take $\delta\leftarrow\delta/\log m$, and apply binary search on $j\in\{1,2,\ldots,m\}$, i.e., apply the algorithm to $j\in\{1,\ldots,\lfloor m/2\rfloor\}$ and $j\in\{\lceil m/2\rceil,\ldots,m\}$ respectively, and if the output is yes then call the algorithm recursively. This gives an extra $\poly(\log m)$ overhead on both sample complexity and gate complexity, which are still $(\log 1/\delta)\cdot\tilde{O}(\mathcal{S}_{\tr}(B,\epsilon))$ and $\log 1/\delta\cdot\tilde{O}(\sqrt{m}\mathcal{T}_{\tr}(B,\epsilon))$, respectively.
\end{proof}

%=============================================================================
\subsection{Quantum SDP solvers with quantum inputs}
We now instantiate \algo{matrixMW-gain} to the fully quantum version (\algo{efficient-SDP}). A key difference is that we use \defn{quantum-Gibbs} to generate (many copies) of the Gibbs state $\rho^{(t)}$ and rely on \lem{oracle-implementation} to implement \ora{violation}.
At a high-level, the correctness of \algo{efficient-SDP} still roughly comes from \thm{SDP-master-1}, as well as \lem{oracle-implementation}. However, its gate complexity will be efficient because of the help of \ora{trace}, \ora{prep}, and \ora{a}.

\begin{theorem}\label{thm:SDP-feasibility-testing}
Assume we are given \ora{trace}, \ora{prep}, and \ora{a}. Furthermore, assume $\Tr[A_{j}^{+}]+\Tr[A_{j}^{-}]\leq B$ for some bound $B$, and $A_{j}$ have rank at most $r$ for all $j\in\range{m}$. Then for any $\epsilon>0$, feasibility of the SDP in \eqn{SDP-1}, \eqn{SDP-2}, and \eqn{SDP-3} can be tested by \algo{efficient-SDP} with success probability at least 0.96 and at most $\frac{1}{\epsilon^{2}}\tilde{O}\big(\mathcal{S}_{\tr}(B,\epsilon)\mathcal{T}_{\text{Gibbs}}\big(\frac{r}{\epsilon^{2}},\frac{1}{\epsilon^{2}},\frac{B}{\epsilon},\epsilon\big)+\sqrt{m}\mathcal{T}_{\tr}(B,\epsilon)\big)$ quantum gates and queries to \ora{trace}, \ora{prep}, and \ora{a}.
\end{theorem}

\begin{algorithm}[htbp]
Initialize the weight matrix $W^{(1)}=I_{n}$, and $T=\frac{16\ln n}{\epsilon^{2}}$\;
\For{$t=1,2,\ldots,T$}{
	Prepare $\tilde{O}(\mathcal{S}_{\tr}(B,\epsilon))$ samples of the Gibbs state $\rho^{(t)}=\frac{W^{(t)}}{\Tr[W^{(t)}]}$ by \defn{quantum-Gibbs}\; \label{line:Gibbs-efficient}
    Using these $\tilde{O}(\mathcal{S}_{\tr}(B,\epsilon))$ copies of $\rho^{(t)}$, search for a $j^{(t)}\in\range{m}$ such that $\tr(A_{j^{(t)}} \rho^{(t)})>a_{j^{(t)}}+\epsilon$ by \lem{oracle-implementation} with $\delta=\frac{\epsilon^{2}}{400\ln n}$. Take $M^{(t)}=\frac{1}{2}(I_{n}-A_{j^{(t)}})$ if such $j^{(t)}$ is found; otherwise, claim that $\mathcal{S}_{\epsilon}\neq\emptyset$ (the SDP is feasible)\; \label{line:observe-gain-matrix-efficient}
	Define the new weight matrix: $W^{(t+1)}=\exp[-\frac{\epsilon}{2}\sum_{\tau=1}^{t}M^{(\tau)}]$\; \label{line:Gibbs-efficient-weight}
	}
Claim that $\mathcal{S}_{0}=\emptyset$ and terminate the algorithm.
\caption{Efficiently testing the feasibility of SDPs: Quantum input model.}
\label{algo:efficient-SDP}
\end{algorithm}

\begin{proof}[Proof of \thm{SDP-feasibility-testing}]
The correctness of \algo{efficient-SDP} is automatically established by \thm{SDP-master-1}; it suffices to analyze the gate cost of \algo{efficient-SDP}.

In Line \ref{line:Gibbs-efficient} of \algo{efficient-SDP} we apply \defn{quantum-Gibbs} to compute the Gibbs state $\rho^{(t)}$. In round $t$, because $M_{j}=\frac{1}{2}[I_{n}-(A_{j}^{+}-A_{j}^{-})]=\frac{1}{2}I_{n}+\frac{1}{2}A_{j}^{-}-\frac{1}{2}A_{j}^{+}\ \forall\,j\in\range{m}$, we take $K_{t}^{+}=\frac{\epsilon}{2}\sum_{\tau=1}^{t}\frac{1}{2}A_{j^{(\tau)}}^{+}$ and $K_{t}^{-}=\frac{\epsilon}{2}\sum_{\tau=1}^{t}\frac{1}{2}A_{j^{(\tau)}}^{-}$. Because $t\leq\frac{16\ln n}{\epsilon^{2}}$, $K_{t}^{+}$, $K_{t}^{-}$ have rank at most $t\cdot r=O(\log n\cdot r/\epsilon^{2})$, and $\Tr[K_{t}^{+}]$, $\Tr[K_{t}^{-}]$ are at most $\frac{\epsilon t}{4}\cdot B=O(\log n\cdot B/\epsilon)$, \defn{quantum-Gibbs} guarantees that
\begin{align}
\mathcal{T}_{\text{Gibbs}}\Big(\frac{r\log n}{\epsilon^{2}},\frac{16\ln n}{\epsilon^{2}},\frac{B\log n}{\epsilon},\epsilon\Big)=\tilde{O}\Big(\mathcal{T}_{\text{Gibbs}}\Big(\frac{r}{\epsilon^{2}},\frac{1}{\epsilon^{2}},\frac{B}{\epsilon},\epsilon\Big)\Big)
\end{align}
quantum gates and queries to \ora{trace}, \ora{prep}, and \ora{a} suffice to prepare the Gibbs state $\rho^{(t)}$. Because there are at most $\frac{16\ln n}{\epsilon^{2}}$ iterations and in each iteration $\rho^{(t)}$ is prepared for $\tilde{O}(\mathcal{S}_{\tr}(B,\epsilon))$ copies, in total the gate cost for Gibbs state preparation is
\begin{align}\label{eqn:cost-quantum-1}
\frac{16\ln n}{\epsilon^{2}}\cdot\tilde{O}\Big(\mathcal{S}_{\tr}(B,\epsilon)\mathcal{T}_{\text{Gibbs}}\Big(\frac{r}{\epsilon^{2}},\frac{1}{\epsilon^{2}},\frac{B}{\epsilon},\epsilon\Big)\Big).
\end{align}

Furthermore, by \lem{oracle-implementation}, Line \ref{line:observe-gain-matrix-efficient} finds a $j^{(t)}\in\range{m}$ such that $\tr(A_{j^{(t)}} \rho^{(t)})>a_{j^{(t)}}+\epsilon$ with success probability at least $1-\frac{\epsilon^{2}}{400\ln n}$, using $\tilde{O}(\sqrt{m}\mathcal{T}_{\tr}(B,\epsilon))$ quantum gates and queries to \ora{trace}, \ora{prep}, and \ora{a}. Because \algo{efficient-SDP} has at most $\frac{16\ln n}{\epsilon^{2}}$ iterations, with success probability at least $1-\frac{16\ln n}{\epsilon^{2}}\cdot\frac{\epsilon^{2}}{400\ln n}=0.96$ we can assume that \lem{oracle-implementation} works correctly, and the total cost of running Line \ref{line:observe-gain-matrix-efficient} is
\begin{align}\label{eqn:cost-quantum-2}
\frac{16\ln n}{\epsilon^{2}}\cdot\tilde{O}\big(\sqrt{m}\mathcal{T}_{\tr}(B,\epsilon)\big).
\end{align}

In total, by \eqn{cost-quantum-1} and \eqn{cost-quantum-2}, the gate complexity of executing \algo{efficient-SDP} is
\begin{align}
&\frac{16\ln n}{\epsilon^{2}}\cdot\tilde{O}\Big(\mathcal{S}_{\tr}(B,\epsilon)\mathcal{T}_{\text{Gibbs}}\Big(\frac{r}{\epsilon^{2}},\frac{1}{\epsilon^{2}},\frac{B}{\epsilon},\epsilon\Big)\Big)+\frac{16\ln n}{\epsilon^{2}}\cdot\tilde{O}\big(\sqrt{m}\mathcal{T}_{\tr}(B,\epsilon)\big) \nonumber \\
&\quad=\frac{1}{\epsilon^{2}}\tilde{O}\Big(\mathcal{S}_{\tr}(B,\epsilon)\mathcal{T}_{\text{Gibbs}}\Big(\frac{r}{\epsilon^{2}},\frac{1}{\epsilon^{2}},\frac{B}{\epsilon},\epsilon\Big)+\sqrt{m}\mathcal{T}_{\tr}(B,\epsilon)\Big).
\end{align}
\end{proof}

To be more explicit, in later sections we prove that:
\begin{itemize}
\item \lem{oracle-implementation-SWAP}: $\mathcal{S}_{\tr}(B,\epsilon)=\mathcal{T}_{\tr}(B,\epsilon)=O(B^{2}\log m/\epsilon^{2})$.
\item \lem{gibbs_prep}: $\mathcal{T}_{\text{Gibbs}}(r_{K},\Phi,B_{K},\epsilon)=O(\Phi\cdot\poly(\log n, r_{K}, B_{K}, \epsilon^{-1}))$.
\end{itemize}

As a consequence, we have the following complexity result for solving SDPs under the quantum input model:

\begin{corollary}\label{cor:SDP-feasibility-testing-quantum}
Assume we are given \ora{trace}, \ora{prep}, and \ora{a}. Furthermore, assume $\Tr[A_{j}^{+}]+\Tr[A_{j}^{-}]\leq B$ for some bound $B$, and $A_{j}$ have rank at most $r$ for all $j\in\range{m}$. Then for any $\epsilon>0$, feasibility of the SDP in \eqn{SDP-1}, \eqn{SDP-2}, and \eqn{SDP-3} can be tested by \algo{efficient-SDP} with success probability at least 0.96 and at most $(\sqrt{m}+\poly(r))\cdot\poly(\log m,\log n,B,\epsilon^{-1})$ quantum gates and queries to \ora{trace}, \ora{prep}, and \ora{a}.
\end{corollary}

\begin{proof}
By \thm{SDP-feasibility-testing}, the complexity of solving the SDP is
\begin{align}\label{eqn:SDP-feasibility-testing-result}
&\frac{1}{\epsilon^{2}}\tilde{O}\Big(\frac{B^{2}\log m}{\epsilon^{2}}\cdot\frac{1}{\epsilon^{2}}\poly\Big(\log n,\frac{r}{\epsilon},\frac{B}{\epsilon},\frac{1}{\epsilon}\Big)+\sqrt{m}\cdot\frac{B^{2}\log m}{\epsilon^{2}}\Big) \nonumber \\
&\quad=(\sqrt{m}+\poly(r))\cdot\poly(\log m,\log n,B,\epsilon^{-1}).
\end{align}
\end{proof}

\begin{remark}
When we use \defn{quantum-Gibbs} to prepare the Gibbs state $\rho^{(t)}$ in Line \ref{line:Gibbs-efficient} of \algo{efficient-SDP}, we have $W^{(t)}=\exp[-\frac{\epsilon t}{4}I_{n}+K_{t}^{+}-K_{t}^{-}]$ by Line \ref{line:Gibbs-efficient-weight} which actually has an extra $-\frac{\epsilon t}{4}I_{n}$ term. However, for any constant $c\in\R$ and Hermitian matrix $H$ we have
\begin{align}
\frac{e^{cI-H}}{\Tr[e^{cI-H}]}=\frac{e^{c}e^{-H}}{\Tr[e^{c}e^{-H}]}=\frac{e^{-H}}{\Tr[e^{-H}]},
\end{align}
hence this $-\frac{\epsilon t}{4}I_{n}$ term does not change $\rho^{(t)}$.
\end{remark}

\begin{remark}
In \cor{SDP-feasibility-testing-quantum}, the only restriction on the decomposition $A_{j}=A_{j}^{+}-A_{j}^{-}$ for all $j\in\range{m}$ is that $\Tr[A_{j}^{+}]+\Tr[A_{j}^{-}]\leq B$. If we assume this decomposition to be the eigen-decomposition, i.e., $A_{j}^{+}$ represents the subspace spanned by the eigenvectors of $A_{j}$ with positive eigenvalues, and $A_{j}^{-}$ represents the subspace spanned by the eigenvectors of $A_{j}$ with negative eigenvalues, then by the low-rank assumption and $-I\preceq A_{j}\preceq I$, $\Tr[A_{j}^{+}]+\Tr[A_{j}^{-}]\leq r$. In this case, \cor{SDP-feasibility-testing-quantum} takes at most $\sqrt{m}\cdot\poly(\log m,\log n,r,\epsilon^{-1})$ quantum gates and queries to \ora{trace}, \ora{prep}, and \ora{a}.
\end{remark}

\begin{remark}
The $\sqrt{m}$ dependence is optimal compared to \thm{SDP-feasibility-lower} proved later.
\end{remark}

\begin{remark}
Using more elaborated techniques and analyses, Ref. \cite{vAG18} explicitly computed the degrees of the parameters in \eqn{SDP-feasibility-testing-result} and improved the complexity of \cor{SDP-feasibility-testing-quantum} to $\tilde{O}(\frac{B\sqrt{m}}{\epsilon^{4}}+\frac{B^{3.5}}{\epsilon^{7.5}})$ (the rank $r$ is implicitly contained in $B$ and hence this complexity is independent of $r$).
\end{remark}

%=============================================================================
\subsection{Trace estimation}
In this subsection, we prove:
\begin{lemma}\label{lem:oracle-implementation-SWAP}
Assume we are given \ora{trace}, \ora{prep}, \ora{a}, and $O(B^{2}\log m/\epsilon^{2})$ copies of a state $\rho\in\C^{n\times n}$, where $\Tr[A_{j}^{+}]+\Tr[A_{j}^{-}]\leq B$ for some bound $B$ for all $j\in\range{m}$. Then for any $\epsilon>0$, \algo{SWAP-test} distinguishes whether $\tr(A_{j} \rho)>a_{j}+\epsilon$ or $\tr(A_{j} \rho)\leq a_{j}$ with success probability at least $1-O(1/\poly(m))$. In other words, $\mathcal{S}_{\tr}(B,\epsilon)=\mathcal{T}_{\tr}(B,\epsilon)=O(B^{2}\log m/\epsilon^{2})$.
\end{lemma}

\begin{algorithm}[htbp]
    Using \ora{prep}, apply the \textsf{SWAP} test on $\rho$ and $\frac{A_{j}^{+}}{\Tr[A_{j}^{+}]}$ for $\poly(\log m,\log n,B,\epsilon^{-1})$ times. Denote the frequency of getting 1 to be $\widetilde{p_{j,+}}$\; \label{line:search-variable-swap-1}
    Using \ora{prep}, apply the \textsf{SWAP} test on $\rho$ and $\frac{A_{j}^{-}}{\Tr[A_{j}^{-}]}$ for $\poly(\log m,\log n,B,\epsilon^{-1})$ times. Denote the frequency of getting 1 to be $\widetilde{p_{j,-}}$\; \label{line:search-variable-swap-2}
    Apply \ora{trace} to compute $\Tr[A_{j}^{+}]$ and $\Tr[A_{j}^{-}]$. Claim that $\tr(A_{j} \rho)>a_{j}+\epsilon$ if $\big(2\widetilde{p_{j,+}}-1\big)\Tr[A_{j}^{+}]-\big(2\widetilde{p_{j,-}}-1\big)\Tr[A_{j}^{-}]>a_{j}+\epsilon/2$, and claim that $\tr(A_{j} \rho)<a_{j}$ if $\big(2\widetilde{p_{j,+}}-1\big)\Tr[A_{j}^{+}]-\big(2\widetilde{p_{j,-}}-1\big)\Tr[A_{j}^{-}]\leq a_{j}+\epsilon/2$\; \label{line:search-variable-cost-2}
\caption{Implementation of the POVM $M_{j}$.}
\label{algo:SWAP-test}
\end{algorithm}

\begin{proof}[Proof of \lem{oracle-implementation-SWAP}]
Recall that the \textsf{SWAP} test \cite{buhrman2001quantum} on $\rho$ and $\frac{A_{j}^{+}}{\Tr[A_{j}^{+}]}$ outputs 1 with probability $\frac{1}{2}+\frac{ \tr(A_{j}^{+} \rho)}{2\Tr[A_{j}^{+}]}$, and the \textsf{SWAP} test on $\rho$ and $\frac{A_{j}^{-}}{\Tr[A_{j}^{-}]}$ outputs 1 with probability $\frac{1}{2}+\frac{\tr(A_{j}^{-}\rho)}{2\Tr[A_{j}^{-}]}$. Therefore, by Chernoff's bound and the fact that $\Tr[A_{j}^{+}],\Tr[A_{j}^{-}]\leq B$, we have
\begin{align}
\Pr\Big[\Big|\widetilde{p_{j,+}}-\Big(\frac{1}{2}+\frac{\tr(A_{j}^{+}\rho)}{2\Tr[A_{j}^{+}]}\Big)\Big|\geq\frac{\epsilon}{8\Tr[A_{j}^{+}]}\Big]&\leq \Pr\Big[\Big|\widetilde{p_{j,+}}-\Big(\frac{1}{2}+\frac{ \tr(A_{j}^{+} \rho)}{2\Tr[A_{j}^{+}]}\Big)\Big|\geq\frac{\epsilon}{8B}\Big] \\
&\leq 2e^{-\frac{O(B^{2}\log m/\epsilon^{2})\cdot\epsilon^{2}}{64B^{2}\cdot 2}} \label{eqn:SDP-efficient-Chernoff-0} \\
&\leq O\Big(\frac{1}{\poly(m)}\Big)\label{eqn:SDP-efficient-Chernoff-1}
\end{align}
for a large constant in the big-$O$ in \eqn{SDP-efficient-Chernoff-0}. Similarly,
\begin{align}
\Pr\Big[\Big|\widetilde{p_{j,-}}-\Big(\frac{1}{2}+\frac{ \tr(A_{j}^{-} \rho)}{2\Tr[A_{j}^{-}]}\Big)\Big|\geq\frac{\epsilon}{8\Tr[A_{j}^{-}]}\Big]\leq O\Big(\frac{1}{\poly(m)}\Big).\label{eqn:SDP-efficient-Chernoff-2}
\end{align}
In other words, with probability at least $1-O\big(\frac{1}{\poly(m)}\big)$,
\begin{align}
\big|\big(2\widetilde{p_{j,+}}-1\big)\Tr[A_{}^{+}]- \tr(A_{j}^{+} \rho) \big|\leq\frac{\epsilon}{4},\quad \big|\big(2\widetilde{p_{j,-}}-1\big)\Tr[A_{j}^{-}]- \tr(A_{j}^{-}  \rho) \big|\leq\frac{\epsilon}{4}.
\end{align}
Therefore, if $\tr(A_{j} \rho)=\tr(A_{j}^{+} \rho)- \tr(A_{j}^{-} \rho)>a_{j}+\epsilon$, then with probability at least $1-O\big(\frac{1}{\poly(m)}\big)$,
\begin{align}
\big(2\widetilde{p_{j,+}}-1\big)\Tr[A_{j}^{+}]-\big(2\widetilde{p_{j,-}}-1\big)\Tr[A_{j}^{-}]>a_{j}+\epsilon/2,
\end{align}
which is exactly the first part of Line \ref{line:search-variable-cost-2}. Similarly, we can use Chernoff's bound to prove that if $\tr(A_{j} \rho)\leq a_{j}$, then with probability at least $1-O\big(\frac{1}{\poly(m)}\big)$,
\begin{align}
\big(2\widetilde{p_{j,+}}-1\big)\Tr[A_{j}^{+}]-\big(2\widetilde{p_{j,-}}-1\big)\Tr[A_{j}^{-}]\leq a_{j}+\epsilon/2,
\end{align}
which is the second part of Line \ref{line:search-variable-cost-2}.

Because \algo{SWAP-test} only uses \textsf{SWAP} which only takes $O(1)$ quantum gates, in total we have $\mathcal{S}_{\tr}(B,\epsilon)=\mathcal{T}_{\tr}(B,\epsilon)=O(B^{2}\log m/\epsilon^{2})$.
\end{proof}

%=============================================================================
\subsection{Gibbs state preparation}
\noindent With the access to \ora{trace} and \ora{prep}, the following lemma shows how to prepare two normalized quantum states $K^{\pm}/\Tr[K^{\pm}]$ where $K^{\pm}=\sum_{j \in S} c_j A_j^{\pm}$, $c_j>0$ and $A_j^{\pm}$ refers to either $A_j^{+}$ or $A_j^{-}$.
\begin{lemma} \label{lem:linear_prep}
$K^{\text{sgn}}/\Tr[K^{\text{sgn}}]$ can be prepared by $|S|$ samples to \ora{trace} and one sample to \ora{prep}, for both $\text{sgn}=+$ and $\text{sgn}=-$.
\end{lemma}
\begin{proof}
Consider the following protocol, where we choose all $\pm$ to be $+$ when preparing $K^{+}/\Tr[K^{+}]$, and choose all $\pm$ to be $-$ when preparing $K^{-}/\Tr[K^{-}]$:
\begin{enumerate}
\item For all $j\in S$, sample \ora{trace} to obtain $\Tr[A_j^{\pm}]$;
\item To prepare $K^{\pm}/\Tr[K^{\pm}]$, toss a coin $i\in S$ such that $\Pr[i=j]=\frac{c_{j}\Tr[A_j^{\pm}]}{\sum_{k\in S}c_{k}\Tr[A_k^{\pm}]}$, take one sample of \ora{prep} to obtain $A_{j}^{\pm}/\Tr[A_{j}^{\pm}]$, and output this state.
\end{enumerate}
By symmetry, we only consider the preparation of $K^{\pm}/\Tr[K^{\pm}]$. With probability $\frac{c_{j}\Tr[A_j^{\pm}]}{\sum_{k\in S}c_{k}\Tr[A_k^{\pm}]}$, the output state is $A_{j}^{\pm}/\Tr[A_{j}^{\pm}]$; therefore, in average the density matrix prepared is
\begin{align}
\sum_{j\in S}\frac{c_{j}\Tr[A_j^{\pm}]}{\sum_{k\in S}c_{k}\Tr[A_k^{\pm}]}\cdot \frac{A_{j}^{\pm}}{\Tr[A_{j}^{\pm}]}=\frac{\sum_{j\in S}c_{j}A_j^{\pm}}{\sum_{k\in S}c_{k}\Tr[A_k^{\pm}]}=\frac{K^{\pm}}{\Tr[K^{\pm}]}.
\end{align}
Furthermore,  Step 1 takes $|S|$ samples to \ora{trace}, and Step 2 takes one sample to \ora{prep}; this exactly matches the sample complexity claimed in \lem{linear_prep}.
\end{proof}

Combining \lem{linear_prep} and \thm{low-rank-Gibbs} leads to a lemma that generates the Gibbs state in Line \ref{line:Gibbs} of \algo{matrixMW-gain}:
\begin{lemma} \label{lem:gibbs_prep}
Suppose $K = K^+ - K^-$, where $K^{\pm}=\sum_{j \in S} c_j A_j^{\pm}$, $c_j>0$ and $A_j^{\pm}$ refers to either $A_j^{+}$ or $A_j^{-}$. Moreover, assume that $\tr(K^+) + \tr(K^-) \leq B_{K}$ for some bound $B_{K}$, and that $K^+$, $K^-$ have rank at most $r_{K}$. Then it is possible to prepare the Gibbs state $\rho_G=\exp(-K)/\tr(\exp(-K))$ to $\epsilon$ precision in trace distance, with $|S|\cdot\poly(\log n, r_{K}, B_{K}, \epsilon^{-1})$ quantum gates and queries to \ora{trace} and \ora{prep}. In other words, $\mathcal{T}_{\text{Gibbs}}(r_{K},\Phi,B_{K},\epsilon)=O(\Phi\cdot\poly(\log n, r_{K}, B_{K}, \epsilon^{-1}))$.
\end{lemma}

%=============================================================================
\subsection{Lower bound for quantum SDP solvers with quantum inputs}
In this section, we prove quantum lower bounds in the quantum input setting.

\begin{theorem}[Lower bound on \thm{SDP-feasibility-testing}]\label{thm:SDP-feasibility-lower}
There exists an SDP feasibility testing problem such that $B,r,\epsilon=\Theta(1)$, and solving the problem requires $\Omega(\sqrt{m})$ calls to \ora{trace} and \ora{prep}.
\end{theorem}

\begin{proof}
Consider the following two instances of the SDP feasibility testing problem:
\begin{enumerate}
\item For all $j\in\range{m}$, set $A_{j}^{-}=0$. For a random $i^{*}\in\range{n}$, set $(A_{j}^{+})_{i^{*}i^{*}}=1$ for all $j\in\range{m}$. All other elements of matrices $A_{j}^{+}$ are set to zero. For a random $j^{*}\in\range{m}$, set $a_{j^{*}}=-1/2$. Set $a_{j}=1/2$ for all $j\neq j^{*}$. Set $\epsilon=1/4$.
\item For all $j\in\range{m}$, set $A_{j}^{-}=0$. For a random $i^{*}\in\range{n}$, set $(A_{j}^{+})_{i^{*}i^{*}}=1$ for all $j\in\range{m}$. All other elements of matrices $A_{j}^{+}$ are set to zero. Set $a_{j}=1/2$ for all $j\in\range{m}$. Set $\epsilon=1/4$.
\end{enumerate}
Note that the first problem is not feasible because there is no $X$ such that $X\succeq 0$ and $\Tr[A_{j^{*}}X]\leq -1/2+1/4<0$; the second problem is always feasible. For both problems, we have $B=r=1$, and the state $\frac{A_{j}^{+}}{\Tr[A_{j}^{+}]}$ is always $|i^{*}\>\<i^{*}|$ for all $j\in\range{m}$. Therefore, \ora{prep} provides no information for distinguishing between the two problems, and we should only rely on \ora{trace}. But this is equivalent to searching for the $j^{*}$ such that $a_{j^{*}}=-1/2$, and by reduction to the lower bound on Grover search it takes at least $\Omega(\sqrt{m})$ queries to \ora{trace} for distinguishing between the two problems.
\end{proof}

Combining \thm{SDP-feasibility-testing} and \thm{SDP-feasibility-lower}, we obtain the optimal bound on SDP feasibility testing using \ora{trace} and \ora{prep}, up to poly-logarithmic factors.

%%%%%%%%%%%%%%%%%%%%%%%%%%%%%%%%%%%%%%%%%%%%%%%%%%%%%%%%%%%%%%%%%%%%%%%%%%%%%%

\section{Application: efficient learnability of quantum states} \label{append:learn}

We consider the following quantum state learning problem, also named "shadow tomography" in \cite{aaronson2017quantum}.

\begin{question}\label{ques:shadow-tomography}
Let $\rho$ be an unknown quantum state in an $n$-dimensional Hilbert space, $E_{1},\ldots,E_{m}$ be known two-outcome POVMs, and $0<\epsilon<1$. Given independent copies of $\rho$, one wants to obtain an explicit quantum circuit for a state $\sigma$ such that with probability at least $2/3$, $|\Tr[\sigma E_{i}]-\Tr[\rho E_{i}]|\leq\epsilon\ \forall\,i\in\range{m}$. What is the sample complexity (i.e., the number of required copies of $\rho$) and gate complexity (i.e., the total running time) of the best such procedure?
\end{question}

Aaronson provides a solution with the sample complexity (i.e., the number of copies of $\rho$) of $\tilde{O}\big(\log^{4}m\cdot\log n/\epsilon^{5}\big)$ in~\cite{aaronson2017quantum}. In this section we show that, for low rank matrices and small $m$, we can also make the learning process \textit{computationally efficient} while keeping a comparable sample complexity,  by using our previous result on speeding up solutions to SDPs.

%=============================================================================
\subsection{Reduction of \ques{shadow-tomography} to SDP feasibility}
We start with a simple explanation of using the solution to SDP feasibility to address \ques{shadow-tomography}. Given (many copies of) any unknown quantum state $\rho$ and two-outcome POVMs $E_{1},\ldots,E_{m}$, in order to estimate $\Tr[\rho E_{i}]$, it suffices to find a state $\sigma$ that is the solution to the following SDP feasibility problem:

\begin{align}
\Tr[\sigma E_{i}]&\leq\Tr[\rho E_{i}]+\epsilon\quad\forall\,i\in\range{m}; \label{eqn:SDP-shadow-1} \\
\Tr[\sigma E_{i}]&\geq\Tr[\rho E_{i}]-\epsilon\quad\forall\,i\in\range{m}; \label{eqn:SDP-shadow-2} \\
\Tr[\sigma]&=1; \label{eqn:SDP-shadow-3} \\
\sigma&\succeq 0. \label{eqn:SDP-shadow-4}
\end{align}

\noindent Any feasible solution $\sigma$ satisfies that $|\Tr[\sigma E_{i}]-\Tr[\rho E_{i}]|<\epsilon$ for all $i\in\range{m}$.
Thus, our quantum SDP solver will generate a description of such $\sigma$. However, we do not know $\Tr[\rho E_{i}]$, and hence the constraints of the SDP feasibility problem, in advance.
The \emph{key} observation is that our SDP solver only relies on the implementation of \ora{violation},
which does not need the knowledge of $\Tr[\rho E_{i}]$ for each $i$ explicitly.
It turns out that with the help of the fast quantum OR lemma, one only needs a few copies of $\rho$ for the implementation of \ora{violation}.

%=============================================================================
\subsection{Finding the violated constraint using $\tilde{O}(\sqrt{m})$ gates}
Similar to \append{SDP-efficient}, we assume the existence of \ora{trace} and \ora{prep} to achieve efficient quantum circuits. Specifically, for the feasibility problem \eqn{SDP-shadow-1}-\eqn{SDP-shadow-4}, we have:

\hd{\ora{trace} for traces of $E_{i}$:} A unitary $O_{\Tr}$ such that for any $i\in\range{m}$, $O_{\Tr}|i\>|0\>=|i\>|\Tr[E_{i}]\>$.

\hd{\ora{prep} for preparing $E_{i}$:} A unitary $O$ (and its inverse $O^\dagger$) acting on $\C^{m}\otimes(\C^{n} \otimes \C^{n})$ such that for any $i\in\range{m}$,
\begin{align}
O|i\>\<i|\otimes |0\>\<0|O^{\dagger}=|i\>\<i|\otimes |\psi_i\>\<\psi_i|,
\end{align}
where $|\psi_i\>\<\psi_i|$ is any purification of $\frac{E_{i}}{\Tr[E_{i}]}$.

Furthermore, we assume that the POVM operator $E_{i}$ has rank at most $r$, for all $i\in\range{m}$. Using \ora{trace} and \ora{prep}, \ora{violation} (searching for violation) can be implemented by the following lemma:
\begin{lemma}\label{lem:oracle-implementation-learning}
Given $\epsilon,\delta\in (0,1)$. Assume we have \ora{trace}, \ora{prep}, and $\poly(\log m,\log n,r,\epsilon^{-1},\log\delta^{-1})$ copies of two states $\rho,\sigma\in\C^{n}$. Assume either $\exists\,i \in [m]$ such that $|\Tr[\sigma E_{i}]-\Tr[\rho E_{i}]|\geq\epsilon$, or  $|\Tr[\sigma E_{i}]-\Tr[\rho E_{i}]|\leq\epsilon/2$ for all $i\in\range{m}$. Then there is an algorithm that in the former case, finds such an $i$; and in the latter case, returns "FEASIBLE". This algorithm has success probability $1-\delta$ and uses in total $\sqrt{m}\cdot\poly(\log m,\log n,r,\epsilon^{-1},\log\delta^{-1})$ quantum gates and queries to \ora{trace} and \ora{prep}.
\end{lemma}

\lem{oracle-implementation-learning} also follows from our fast quantum OR lemma (\lem{fast-quantum-OR-cite}) by combining \lem{oracle-implementation-SWAP} and \lem{oracle-implementation}, where we replace $\rho$ by $\rho^{\otimes C} \otimes \sigma^{\otimes C} \otimes \proj{0}^a$ for some $C=\poly(\log m,\log n,\epsilon^{-1})$; also notice that because $0\leq E_{i}\leq I$ and $\rank(E_{i})\leq r$, we have $B\leq r$. As a result, the detailed proof is omitted here.

%=============================================================================
\subsection{Gate complexity of learning quantum states}
Similar to \cor{SDP-feasibility-testing-quantum}, we solve \ques{shadow-tomography} by using \lem{gibbs_prep} to generate (copies) of the Gibbs state $\rho^{(t)}$ and relying on \lem{oracle-implementation-learning} to implement \ora{violation}.

\begin{corollary}\label{cor:efficient-shadow}
Assume we are given \ora{trace} and \ora{prep}. Then for any $\epsilon>0$, \ques{shadow-tomography} can be solved by \algo{efficient-shadow} with success probability at least 0.96, using at most $\poly(\log m,\log n,r,\epsilon^{-1})$ copies of $\rho$, and at most $\sqrt{m}\cdot\poly(\log m,\log n,r,\epsilon^{-1})$ quantum gates and queries to \ora{trace} and \ora{prep}.
\end{corollary}

\begin{algorithm}[htbp]
Initialize the weight matrix $W^{(1)}=I_{n}$, and $T=\frac{16\ln n}{\epsilon^{2}}$\;
\For{$t=1,2,\ldots,T$}{
	Prepare $\poly(\log m,\log n,r,\epsilon^{-1})$ samples of the Gibbs state $\rho^{(t)}=\frac{W^{(t)}}{\Tr[W^{(t)}]}$ by \lem{gibbs_prep}, and take $\poly(\log m,\log n,r,\epsilon^{-1})$ copies of $\rho$\; \label{line:Gibbs-efficient-shadow}
  Using these $\poly(\log m,\log n,r,\epsilon^{-1})$ copies of $\rho^{(t)}$ and $\rho$, apply \lem{oracle-implementation-learning} with $\delta=\frac{\epsilon^{2}}{400\ln n}$ to search for an $i^{(t)}\in\range{m}$ such that $|\Tr[\rho E_{i^{(t)}}]-\Tr[\rho^{(t)} E_{i^{(t)}}]|\geq\epsilon$. \label{line:observe-gain-matrix-shadow} \eIf{such $i^{(t)}$ is found}
    {\eIf{$(\Tr[\rho E_{i^{(t)}}]-\Tr[\rho^{(t)} E_{i^{(t)}}]\geq\epsilon)$}{Take $M^{(t)}=\frac{1}{2}\big(I_{n}-(-E_{i^{(t)}}+\Tr[\rho E_{i^{(t)}}]I_{n})\big)$\; \label{line:shadow-case-1}}({$(\Tr[\rho E_{i^{(t)}}]-\Tr[\rho^{(t)} E_{i^{(t)}}]\leq-\epsilon)$}){Take $M^{(t)}=\frac{1}{2}\big(I_{n}-(E_{i^{(t)}}-\Tr[\rho E_{i^{(t)}}]I_{n})\big)$\; \label{line:shadow-case-2}}}
    ({(no such $i^{(t)}$ exists)}){Claim $\rho^{(t)}$ to be the solution, and terminate the algorithm\; \label{line:observe-gain-matrix-efficient-shadow}}
	Define the new weight matrix: $W^{(t+1)}=\exp[-\frac{\epsilon}{2}\sum_{\tau=1}^{t}M^{(\tau)}]$\;
	}
\caption{Efficiently learn a quantum state via measurements.}
\label{algo:efficient-shadow}
\end{algorithm}

\begin{proof}
Similar to \cor{SDP-feasibility-testing-quantum}, the correctness of \algo{efficient-shadow} is automatically established by \thm{SDP-master-1}; it suffices to analyze the gate cost of \algo{efficient-shadow}.

In Line \ref{line:Gibbs-efficient-shadow} of \algo{efficient-shadow} we apply \lem{gibbs_prep} to compute the Gibbs state $\rho^{(t)}$. In round $t$, because either
\begin{align}\label{eqn:shadow-case-1}
M^{(t)}=\frac{1}{2}\big(I_{n}-(-E_{i^{(t)}}+\Tr[\rho E_{i^{(t)}}]I_{n})\big)=\frac{1-\Tr[\rho E_{i^{(t)}}]}{2}I_{n}+\frac{1}{2}E_{i^{(t)}}
\end{align}
when Line \ref{line:shadow-case-1} executes, or
\begin{align}\label{eqn:shadow-case-2}
M^{(t)}=\frac{1}{2}\big(I_{n}-(E_{i^{(t)}}-\Tr[\rho E_{i^{(t)}}]I_{n})\big)=\frac{1+\Tr[\rho E_{i^{(t)}}]}{2}I_{n}-\frac{1}{2}E_{i^{(t)}}
\end{align}
when Line \ref{line:shadow-case-2} executes, we can take $K_{t}^{+}=\frac{\epsilon}{2}\sum_{\tau=1}^{t}\frac{1}{2}E_{i^{(\tau)}}^{+}$ and $K_{t}^{-}=\frac{\epsilon}{2}\sum_{\tau=1}^{t}\frac{1}{2}E_{i^{(\tau)}}^{-}$, where $E_{i^{(\tau)}}^{+}=E_{i^{(\tau)}}$, $E_{i^{(\tau)}}^{-}=0$ when \eqn{shadow-case-1} holds for round $\tau$, and $E_{i^{(\tau)}}^{+}=0$, $E_{i^{(\tau)}}^{-}=E_{i^{(\tau)}}$ when \eqn{shadow-case-2} holds for round $\tau$. Because $t\leq\frac{16\ln n}{\epsilon^{2}}$, $K_{t}^{+}$, $K_{t}^{-}$ have rank at most $t\cdot r=O(\log n\cdot r/\epsilon^{2})$ and $\Tr[K_{t}^{+}]$, $\Tr[K_{t}^{-}]$ are at most $\frac{\epsilon t}{4}\cdot r=O(\log n\cdot r/\epsilon)$, \lem{gibbs_prep} guarantees that
\begin{align}
\frac{16\ln n}{\epsilon^{2}}\cdot\poly\Big(\log n,\frac{r\log n}{\epsilon^{2}},\frac{r\log n}{\epsilon},\epsilon^{-1}\Big)=\poly(\log n,r,\epsilon^{-1})
\end{align}
quantum gates and queries to \ora{trace} and \ora{prep} suffice to prepare the Gibbs state $\rho^{(t)}$. Because there are at most $\frac{16\ln n}{\epsilon^{2}}=\poly(\log n,\epsilon^{-1})$ iterations and in each iteration $\rho^{(t)}$ is prepared for $\poly(\log m,\log n,r,\epsilon^{-1})$ copies, in total the gate cost for Gibbs state preparation is $\poly(\log m,\log n,r,\epsilon^{-1})$.

Furthermore, by \lem{oracle-implementation-learning}, \algo{efficient-shadow} finds an $i^{(t)}\in\range{m}$ such that $|\Tr[\rho E_{i^{(t)}}]-\Tr[\rho^{(t)} E_{i^{(t)}}]|\geq\epsilon$ with success probability at least $1-\frac{\epsilon^{2}}{400\ln n}$, using $\poly(\log m,\log n,r,\epsilon^{-1})$ copies of $\rho$, and $\sqrt{m}\cdot\poly(\log m,\log n,r,\epsilon^{-1})$ quantum gates and queries to \ora{trace} and \ora{prep}. Because \algo{efficient-shadow} has at most $\frac{16\ln n}{\epsilon^{2}}$ iterations, with success probability at least $1-\frac{16\ln n}{\epsilon^{2}}\cdot\frac{\epsilon^{2}}{400\ln n}=0.96$ we can assume that the quantum search in \lem{oracle-implementation-learning} works correctly, and the total gate cost of calling \algo{efficient-shadow} is $\sqrt{m}\cdot\poly(\log m,\log n,r,\epsilon^{-1})$.

In conclusion, $\poly(\log m,\log n,r,\epsilon^{-1})$ is an upper bound on the number of copies of $\rho$, and $\sqrt{m}\cdot\poly(\log m,\log n,r,\epsilon^{-1})$ is an upper bound on the total number of quantum gates and queries to \ora{trace} and \ora{prep}.
\end{proof}

\begin{remark}\label{rem:lower_bound_learn}
Using the same idea as \thm{SDP-feasibility-lower}, we can prove that there exists a shadow tomography problem such that $r,\epsilon=\Theta(1)$, and solving the problem requires $\Omega(\sqrt{m})$ calls to \ora{trace} and \ora{prep}. Therefore \cor{efficient-shadow} is also optimal up to poly-logarithmic factors.
\end{remark}

%%%%%%%%%%%%%%%%%%%%%%%%%%%%%%%%%%%%%%%%%%%%%%%%%%%%%%%%%%%%%%%%%%%%%%%%%%%%%%%

\section{Gibbs sampling of low-rank Hamiltonians} \label{append:low-rank_sampling}
In this section, we demonstrate how to sample from the Gibbs state of low-rank Hamiltonians given a quantum oracle generating desired states. We repeatedly use the following result of \cite{lloyd2013quantum} (with a straightforward generalization in \cite{kimmel2017hamiltonian}):
\begin{lemma}[\cite{lloyd2013quantum, kimmel2017hamiltonian}]\label{lem:HamiltonianSimulation}
Suppose we are given a quantum oracle that prepares copies of two unknown (normalized) $l$-qubit quantum states $\rho^+$ and $\rho^-$, and we wish to evolve under the Hamiltonian $H = a_+ \rho^+ - a_- \rho^-$ for some nonnegative numbers $a_+, a_- \ge 0$. Then we can approximately implement the unitary $\exp(iHt)$ up to diamond-norm error $\delta$, using $O(a^2 t^2 / \delta)$ copies of $\rho^+$ and $\rho^-$ and $O(la^2 t^2 / \delta)$ other 1- or 2-qubit gates, where $a = a_+ + a_-$.
\end{lemma}

By using phase estimation on the operator $\exp(iHt)$ with $t = O(1/a)$, we have

\begin{lemma}
Under the same assumptions as Lemma~\ref{lem:HamiltonianSimulation}, we can perform eigenvalue estimation of $H$: given an eigenstate of $H$, we can estimate its eigenvalue up to precision $\epsilon$, with probability $1 - \xi$, using $O(a^2\epsilon^{-2}\xi^{-2})$ copies of $\rho^+$ and $\rho^-$ and $O(la^2 \epsilon^{-2}\xi^{-2})$ other 1- or 2-qubit gates, where $a = a_+ + a_-$. This procedure disturbs the input state by at most a trace distance error of $O(\sqrt{\xi})$.
\end{lemma}

In the following proof, we instead assume that eigenvalue estimation of $H$ can be done exactly. This assumption is not true, but it helps to simplify the exposition; the assumption will be removed in \append{gibbs}.

\subsection{Computing the partition function}\label{append:partitionFunction}
As a warm-up, we start with the following lemma:
\begin{lemma}\label{lem:partitionFunction}
Suppose $K = K^+ - K^-$, where $K^+$ and $K^-$ are $n \times n$ PSD matrices, and there is a quantum oracle that prepares copies of the states $\rho^+=K^+ / \tr(K^+)$, $\rho^- = K^-/ \tr(K^-)$, and an oracle for the numbers $\tr(K^+)$, $\tr(K^-)$. Moreover, assume that $\tr(K^+) + \tr(K^-) \leq B$ for some bound $B$,\footnote{The $B$ here is denoted as $B_{K}$ in \defn{quantum-Gibbs} and \lem{gibbs_prep}; for simplicity we make this abbreviation throughout \append{low-rank_sampling} and \append{gibbs}.} and that $K^+$, $K^-$ have rank at most $r_{K}$. Then it is possible to estimate the partition function $Z=\tr(\exp(-K))$ to multiplicative error $\epsilon$  with success probability at least $1-\xi$, with $\poly(\log n, r_{K}, B, \epsilon^{-1}, \xi^{-1})$ quantum gates.
\end{lemma}
\begin{proof}[Proof Sketch]
As mentioned above, we assume that we can implement the unitary evolution $\exp(iKt)$ as well as the phase estimation protocol, perfectly to infinite precision. This idealization is made here for the sake of flashing out the core ideas behind the proposed protocol. These assumptions will be lifted in \append{gibbs} (\lem{app:partition}), where a careful error analysis of this scheme is presented.

Under these assumptions, let us first consider the estimation of
\begin{align}\label{eqn:Z_supp}
Z_{\text{supp}} \equiv \sum_{|\lambda_i| \ge \delta} e^{-\lambda_i},
\end{align}
where $0<\delta<1$ is a small threshold and $\lambda_i$'s are eigenvalues of $K$. Since $\delta>0$ is a small parameter, $Z_{\text{supp}}$ is the partition function when considering the approximated support of $K$.

The main idea in the estimation of $Z_{\text{supp}}$ is to perform phase estimation of the unitary operator $e^{2 \pi iK}$ on $\rho^+$ and $\rho^-$, after which we obtain
\begin{equation}\label{eq:barrho}
\rho^\pm=\frac{K^\pm }{ \tr(K^\pm)}\rightarrow \bar{\rho}^\pm=\frac{1}{ \tr(K^\pm)}\sum_{\lambda} \Pi_\lambda K^\pm \Pi_\lambda\otimes\ket{\lambda}\bra{\lambda},
\end{equation}
where $\Pi_\lambda$ is the projection onto the $\lambda$-eigenspace of $K$, and $\lambda$ is any eigenvalue of $K$. Let us define
\begin{equation}\label{eq:Klambda}
K^+_\lambda := \Pi_\lambda K^+ \Pi_\lambda, \quad K^-_\lambda := \Pi_\lambda K^- \Pi_\lambda.
\end{equation}
Then,
\begin{align}\label{eq:Kdiff}
K^+_\lambda - K^-_\lambda = \Pi_\lambda K \Pi_\lambda = \lambda \Pi_\lambda,
\end{align}
and therefore $K^+_\lambda$ and $K^-_\lambda$ differ by a multiple of the identity in their support space (the $\lambda$-eigenspace of $K$). Hence $K^+_\lambda$ and $K^-_\lambda$ are simultaneously diagonalizable, and their corresponding eigenvalues differ by exactly $\lambda$. In other words, there exists an eigenbasis of $K$, which we call $\{\ket{v_i}\}_i$ with corresponding eigenvalues $\lambda_i$, such that $K^+_{\lambda}$ and $K^-_{\lambda}$ are diagonal in this eigenbasis for all $\lambda$. We can therefore write
\begin{equation}\label{eq:Kdiag}
K^+_\lambda = \sum_{i: \lambda_i = \lambda} \lambda^+_i \ket{v_i}\bra{v_i}, \quad \quad
K^-_\lambda = \sum_{i: \lambda_i = \lambda} \lambda^-_i \ket{v_i}\bra{v_i},
\end{equation}
for some nonnegative numbers $\lambda^+_i$, $\lambda^-_i$ satisfying $\lambda^+_i - \lambda^-_i = \lambda_i$. Combining Eqs. ~\eqref{eq:barrho} and~\eqref{eq:Kdiag}, we obtain that $\bar{\rho}^+$ ($\bar{\rho}^-$) -- the state after performing phase estimation of the unitary operator $e^{2 \pi iK}$ on $\rho^+$ ($\rho^-$) is given by
\begin{equation}\label{eq:barrho2}
\bar{\rho}^\pm=\frac{1}{ \tr(K^\pm)}\sum_{\lambda_i} \lambda^{\pm}_i \ket{v_i}\bra{v_i}\otimes\ket{\lambda_i}\bra{\lambda_i}.
\end{equation}

Now consider the following procedure, and let its output be the random variable $X$:

\begin{algorithm}[H]
\caption{Estimation of $Z_{\text{supp}}$}
 \label{algo:Zsupp}
\begin{enumerate}
\item Let $\text{sgn} = +$ with probability $\tr(K^+)/[\tr(K^+)+\tr(K^-)]$, and $\text{sgn} = -$ otherwise.
\item Perform phase estimation of the operator $e^{2 \pi iK}$ on $\rho^{\text{sgn}}$; Let the output state be $\bar{\rho}^{\text{sgn}}=\frac{1}{ \tr(K^{\text{sgn}})}\sum_{\lambda_i} \lambda^{\text{sgn}}_i \ket{v_i}\bra{v_i}\otimes\ket{\lambda_i}\bra{\lambda_i}$. Measure the second register and let the obtained eigenvalue of $K$ be $\lambda$.
\item If $|\lambda| < \delta$ output 0; else if $\text{sgn}=+$ output $\lambda^{-1}e^{-\lambda}$; else output $-\lambda^{-1}e^{-\lambda}$.
\end{enumerate}
\end{algorithm}
Then, under the assumption of perfect phase estimation, we have
\begin{align}
\E[X] &=  \frac{\tr(K^+)}{\tr(K^+)+\tr(K^-)}\sum_{|\lambda_i| \ge \delta} \frac{\lambda_i^+}{\tr(K^+)} \frac{e^{-\lambda_i}}{\lambda_i}-\frac{\tr(K^-)}{\tr(K^+)+\tr(K^-)}\sum_{|\lambda_i| \ge \delta} \frac{\lambda_i^-}{\tr(K^-)} \frac{e^{-\lambda_i}}{\lambda_i} \nonumber\\ &=\frac{1}{\tr(K^+)+\tr(K^-)}\sum_{|\lambda_i| \ge \delta} e^{-\lambda_i} = \frac{Z_{\text{supp}}}{\tr(K^+)+\tr(K^-)},
\end{align}
where $\lambda_i^\pm$ are the eigenvalues of  $K^\pm_{\lambda_i}$, satisfying $\lambda_i^+ - \lambda_i^-=\lambda_i$. Therefore $\E[X]$ is proportional to $Z_{\text{supp}}$, and obtaining a multiplicative estimate of $\E[X]$ gives us a multiplicative estimate of $Z_{\text{supp}}$.

The second moment of $X$ reads
\begin{equation}
\E[X^2] =\frac{1}{\tr(K^+)+\tr(K^-)}\sum_{|\lambda_i| \ge \delta} (\lambda_i^+ + \lambda_i^-) \frac{e^{-2\lambda_i}}{\lambda_i^2}\le \max_{|\lambda_i| \ge \delta} |\lambda_i|^{-2} e^{-2\lambda_i} \le \delta^{-2} Z^2_{\text{supp}}.
\end{equation}
 We see that $\E[X^2] \le B^2\delta^{-2} \E[X]^2$, and therefore by Chebyshev's inequality we can obtain, with constant probability, an $\epsilon$-error multiplicative estimate of $\E[X]$, hence of $Z_{\text{supp}}$, by running the above procedure $O(B^2\delta^{-2}\epsilon^{-2})$ times and taking the mean.

We still need to calculate $Z$, the full partition function including small eigenvalues of $K$. Let $R$ denote the number of eigenvalues of $K$ (including degeneracy) with absolute value at least $\delta$, and note that $R \le 2r_{K}$, where recall that $r_{K}$ upper bounds the rank of $K^+$ and $K^-$ . Define the following approximation of $Z$:
\begin{equation}
Z' \equiv Z_{\text{supp}} + (n-R) = \sum_{|\lambda_i| \ge \delta} e^{-\lambda_i} + \sum_{|\lambda_i| < \delta} e^{0}.
\end{equation}
Using $e^\delta \le 1+2\delta$ and $e^{-\delta} \ge 1-\delta$, we get that
\begin{equation}
| Z - Z' | \le 2\delta(n-R).
\end{equation}
Therefore if we make $\delta$ small enough,  say $\delta = O(\epsilon)$, $Z'$ gives a good multiplicative estimate for $Z$.

To compute $Z'$, we need a good multiplicative estimate of $n-R$. This can essentially be done by estimating the probability of a random state having eigenvalue smaller than $\delta$. Let the output of the following procedure be $Y$:

\begin{algorithm}[H]
\caption{Estimation of $n-R$}
 \label{algo:n-R}
\begin{enumerate}
\item Perform phase estimation of the operator $e^{2\pi iK}$ on the uniformly random state $I/n$; let the output eigenvalue be $\lambda$.
\item If $|\lambda| < \delta$ output 1; otherwise output 0.
\end{enumerate}
\end{algorithm}
$Y$ is a Bernoulli random variable with mean $\E[Y] = (n-R)/n$ and variance $\text{Var}[Y] = R (n-R) / n^2 \leq R \E[Y]^2$. By Chebyshev's inequality, $O(r_{K} \epsilon^{-2})$ repetitions of the above procedure gives us an $\epsilon$-error multiplicative estimate of $\E[Y]$, and thus of $n-R$.

Putting everything together, we see that  $O(B^2\epsilon^{-4} + r_{K}\epsilon^{-2})$ uses of (perfect) phase estimation   of $e^{2\pi iK}$ suffices to get a $O(\epsilon)$-error multiplicative estimate of $Z$, completing the proof.
\end{proof}

\subsection{Sampling from the Gibbs state}\label{append:low-rank-Gibbs}

\begin{theorem}[Full proof deferred to \append{gibbs}]\label{thm:low-rank-Gibbs}
Suppose $K = K^+ - K^-$, where $K^+$ and $K^-$ are $n \times n$ PSD matrices, and there is a quantum oracle that prepares copies of the states $\rho^+=K^+ / \tr(K^+)$, $\rho^- = K^-/ \tr(K^-)$, and an oracle for the numbers $\tr(K^+)$, $\tr(K^-)$. Moreover, assume that $\tr(K^+) + \tr(K^-) \le B$ for some bound $B$, and that $K^+$, $K^-$ have rank at most $r_{K}$. Then it is possible to prepare the Gibbs state $\rho_G=\exp(-K)/\tr(\exp(-K))$ to $\epsilon$ precision in trace distance, with $\poly(\log n, r_{K}, B, \epsilon^{-1})$ quantum gates.
\end{theorem}

\begin{proof}[Proof sketch]
Similar to the proof sketch of the partition function, here as well we assume an infinite precision implementation of the unitary evolution operator $\exp(iKt)$ as well as  of the phase estimation protocol. In addition we assume that quantum principal component analysis can be implemented perfectly.  These assumptions will be lifted in \append{gibbs} (\thm{gibbs}), where a complete proof is presented.

The procedure is somewhat similar to that of calculating the partition function above. We pick $\delta = O(\epsilon)$, a small threshold, and first consider a procedure to sample from
\begin{equation}
\rho_{\text{supp}} \equiv \sum_{|\lambda_i| \ge \delta} e^{-\lambda_i}\ket{v_i}\bra{v_i}/Z_{\text{supp}},
\end{equation}
where $\lambda_i$'s and $\ket{v_i}$'s are eigenvalues and eigenstates of $K$. (In the case that $\rho_{\text{supp}}$ is undefined, i.e. that all eigenvalues of $K$ have magnitude less than $\delta$, it is easy to see that the uniformly mixed state $I/n$ is already an $O(\epsilon)$-trace distance error approximation to $\rho_G$. This is the case when $Z_{\text{supp}} = 0$.)  $\rho_{\text{supp}}$ is the Gibbs state when considering only the (approximated) support of $K$. Consider the procedure in \algo{rhosupp}.

\begin{algorithm}
\caption{Estimation of $\rho_{\text{supp}}$}
 \label{algo:rhosupp}
\begin{enumerate}
\item Let $\text{sgn} = +$ with probability $\tr(K^+)/[\tr(K^+)+\tr(K^-)]$, and $\text{sgn} = -$ otherwise.
\item Perform phase estimation of the unitary operator $e^{2 \pi iK}$ on $\rho^{\text{sgn}}$;  let the output state be $\bar{\rho}^{\text{sgn}}=\frac{1}{ \tr(K^{\text{sgn}})}\sum_{\lambda_i} \lambda^{\text{sgn}}_i \ket{v_i}\bra{v_i}\otimes\ket{\lambda_i}\bra{\lambda_i}$.
\item Project $\bar\rho^{\text{sgn}}$ onto $\bar\varrho^{\text{sgn}}=\frac{1}{ \tr(K^{\text{sgn}})}\sum_{\lambda_i:|\lambda_i|\ge\delta} \lambda^{\text{sgn}}_i \ket{v_i}\bra{v_i}\otimes\ket{\lambda_i}\bra{\lambda_i}$.
\item The average state at this stage is
$$
\bar\varrho=\frac{1}{\tr(K^+)+\tr(K^-)}\sum_{\lambda_i:|\lambda_i|\ge\delta} (\lambda^+_i + \lambda^-_i) \ket{v_i}\bra{v_i}\otimes\ket{\lambda_i}\bra{\lambda_i}.
$$
Perform phase estimation of the operator $e^{2\pi iK}$ on $\bar\varrho$; let the measured eigenvalue be $\mu=\lambda^+ + \lambda^-$, and the resulting state be ${\varrho}_{\mu}$.
\item Accept the state ${\varrho}_{\mu}$ with probability $\frac{\delta}{\mu}\frac{(1-\epsilon)e^{-{\lambda_i}}}{Z'_{\text{supp}}}$, for $Z'_{\text{supp}}$ a $\epsilon$-multiplicative error approximation of $Z_{\text{supp}}$ by \algo{Zsupp}.
\end{enumerate}
\end{algorithm}
Note that $\frac{\delta}{{\lambda}^+ +{\lambda}^-}\frac{(1-\epsilon)e^{-{\lambda}}}{Z'_{\text{supp}}} \le 1$ since $\lambda^\pm\ge0$ and by assumption $|\lambda|=|\lambda^+ -\lambda^-|\ge \delta$, and $Z'_{\text{supp}} \leq (1+\epsilon)Z_{\text{supp}} \leq (1+\epsilon)\sum_{|\lambda_i| \ge \delta} e^{-\lambda_i}$ by \eqn{Z_supp}. Moreover assuming that $K$ has at least one eigenvalue with magnitude at least $\delta$, the success probability in Line 5 of \algo{rhosupp} is at least
\begin{align}\label{eqn:Z_supp-prob}
\frac{\delta}{\tr(K^+)+\tr(K^-)}\sum_{\lambda_i:|\lambda_i|\ge\delta} \frac{(1-\epsilon)e^{-\lambda_i}}{Z'_{\text{supp}}} \geq \frac{\delta(1-\epsilon)}{B(1+\epsilon)},
\end{align}
and therefore we can output $\rho_{\text{supp}}$ efficiently by repeating \algo{rhosupp} until success, which takes $O(B/\delta)$ trials in expectation.

Accounting for the randomness in Step 1, at the end of Step 3, we obtain  the mixed state $\bar\varrho$. However, for Gibbs sampling, we should have factors of the form  $e^{-\lambda_i}\ket{{v}_i}\bra{{v}_i}$ instead of $({\lambda}^+_i + {\lambda}^-_i) \ket{{v}_i}\bra{{v}_i}$ that appear in $\bar\varrho$. Therefore, at this stage of the protocol, to accept $\ket{{v}_i}\bra{{v}_i}$ with probability proportional to  $e^{-\lambda_i}/({\lambda}^+_i + {\lambda}^-_i)$, but for that we need to measure  ${\lambda}^+_i + {\lambda}^-_i$. This is done in steps 4 and 5 of the above procedure, which is equivalent to applying
$\sum_{\lambda_i:|\lambda_i|\ge\delta} \frac{\delta}{{\lambda_i}^+ +{\lambda_i}^-}\frac{e^{-{\lambda_i}}}{Z_{\text{supp}}} \ket{v_i}\bra{v_i}\otimes\ket{\lambda_i}\bra{\lambda_i}
$ to $\bar\varrho$. Upon keeping only the first register we obtain
\begin{equation}
\frac{\delta}{\tr(K^+)+\tr(K^-)}\sum_{|\lambda_i| \ge \delta} \frac{e^{-\lambda_i}}{Z_{\text{supp}}}\ket{v_i}\bra{v_i}\propto \frac{1}{Z_{\text{supp}}} \sum_{|\lambda_i| \ge \delta}  e^{-\lambda_i}\ket{v_i}\bra{v_i}\equiv\rho_{\text{supp}},
\end{equation}
where $\rho_{\text{supp}}$ is the Gibbs state when considering only the (approximated) support of $K$.

We still need to calculate $\rho_G$, the full Gibbs state including small eigenvalues of $K$. Recall that $R$ denotes the number of eigenvalues (including degeneracy) of $K$ with absolute value at least $\delta$, and note that $R \leq 2r_{K}$, where $r_{K}$ upper bounds the rank of $K^+$ and $K^-$ . Define the following approximation of $\rho_G$:
\begin{equation}
\rho_G' \equiv \frac{Z_{\text{supp}}}{Z'}\rho_{\text{supp}}+\frac{n-R}{Z'}\rho_{\text{ker}}= \frac{1}{Z'} (\sum_{|\lambda_i| \ge \delta} e^{-\lambda_i}\ket{v_i}\bra{v_i} + \sum_{|\lambda_i| < \delta} \ket{v_i}\bra{v_i}),
\end{equation}
where $\rho_{\text{ker}}=\frac{1}{n-R}\sum_{|\lambda_i| < \delta} \ket{v_i}\bra{v_i}$ is the  uniformly random state on the orthogonal complement of the (approximate) support of $K$.
Then
\begin{align}
\|\rho_G - \rho_G' \|_{\tr} &=\Big|\frac{1}{Z}-\frac{1}{Z'}\Big|\sum_{|\lambda_i|\ge\delta}e^{-\lambda_i}+\sum_{|\lambda_i|<\delta}\Big|\frac{e^{-\lambda_i}}{Z}-\frac{1}{Z'}\Big| \\
&\le\Big|\frac{1}{Z}-\frac{1}{Z'}\Big|\sum_{|\lambda_i|\ge\delta} e^{-\lambda_i}+\left[\sum_{|\lambda_i|<\delta}\Big|\frac{1}{Z}-\frac{1}{Z'}\Big|+\frac{2\delta}{Z}e^{-\lambda_i}\right] \\
&\le\Big|\frac{1}{Z}-\frac{1}{Z'}\Big|Z'+2\delta\le4\delta.
\end{align}
Therefore if we make $\delta$ small enough, $\rho_G'$ gives a good estimate (in trace distance) for $\rho_G$.

To estimate $\rho_{\text{ker}}$, we consider the output of \algo{rhoker} below.
\begin{algorithm}
\caption{Estimation of $\rho_{\text{ker}}$}
 \label{algo:rhoker}
\begin{enumerate}
\item Perform phase estimation of the operator $e^{2\pi iK}$ on the uniformly random state $I/n$; let the output eigenvalue be $\lambda$ and the resulting state be $\Pi_\lambda$.
\item If $|\lambda| \ge \delta$ abort; otherwise, accept the state.
\end{enumerate}
\end{algorithm}

Finally, $\rho_G$ is generated by running the Algorithm~\ref{algo:rhosupp} with probability $\frac{Z_{\text{supp}}}{Z}=\Omega(\frac{\delta}{B})$ (by \eqn{Z_supp-prob}) until we accept  $\rho_{\text{supp}}$, and running the Algorithm~\ref{algo:rhoker} with probability
\begin{align}
\frac{n-R}{Z}\geq\frac{n-2r_{K}}{n}=\Omega(1),
\end{align}
until we accept $\rho_{\text{ker}}$. The detailed analysis is presented in \append{Gibbs-function}.

In the previous subsection we proved that, upon setting $\delta=O(\epsilon)$ we can obtain $Z_{\text{supp}}, {Z}$ and $n-R$ up to an $O(\epsilon)$ multiplicative error with $\poly(\log n, r_{K}, B, \epsilon^{-1})$ quantum gates. Therefore, Using Lemma 7 of \cite{vanApeldoorn2017quantum} we obtain  $\frac{Z_{\text{supp}}}{Z}$ and $\frac{n-R}{Z}$ to $O(\epsilon)$ multiplicative error. This, in turns, implies the with the above procedure we prepare the Gibbs state $\rho_G$ up to error $O(\epsilon)$ in trace distance, with $\poly(\log n, r_{K}, B, \epsilon^{-1})$ quantum gates.
\end{proof}

%%%%%%%%%%%%%%%%%%%%%%%%%%%%%%%%%%%%%%%%%%%%%%%%%%%%%%%%%%%%%%%%%%%%%%%%%%%%%%

\section{Proof for Gibbs sampling of low-rank states}~\label{append:gibbs}
In this section we provide a complete proof of \lem{partitionFunction} and~\thm{low-rank-Gibbs} given in \append{low-rank_sampling}.

\subsection{Preliminaries}
Before we start the proof, we first gather some preliminary facts that we need. First of all, the output of the phase estimation protocol is probabilistic and depends on the measurement. This is a problem since we often want the output of phase estimation to be consistent across multiple runs of the protocol. For example, when the eigenvalue is to our cutoff $\delta$, the phase estimation protocol may not be able to consistently decide whether it was greater than or less than $\delta$. To overcome this problem, in the proofs below we use Ta-Shma's \emph{consistent phase estimation} algorithm instead:
\begin{lemma}[\cite{tashma2013inverting}] \label{lem: consistent}
Let $U$ be a $D$-dimensional unitary matrix, and $\delta, \xi > 0$. There is a quantum algorithm that first chooses a random shift $s$, such that with probability at least $1-D\xi$ the following holds for \emph{all} eigenstates $\ket{v_\lambda}$ of $U$ (where $U\ket{v_\lambda} = e^{2 \pi i \lambda} \ket{v_\lambda}$) (in this case we call $s$ a \emph{good} shift):
\begin{itemize}
\item On input $\ket{v_\lambda}\ket{\bar{0}}$, where $\ket{\bar{0}}$ is a fixed reference state, the algorithm outputs a state $O(\sqrt{\xi})$-close to $\ket{v_\lambda}\ket{f(s,\lambda)}$ in trace distance.
\item $f(s,\lambda)$ is a function only of $s$ and $\lambda$, and $|f(s,\lambda) - \lambda| < \delta$.
\end{itemize}
This algorithm requires $\poly(\xi^{-1},\delta^{-1})$ uses of the controlled-$U$ operation and other quantum gates.
\end{lemma}
The essential idea of this algorithm is to choose a random shift $s$, and perform phase estimation on $e^{is} U$ instead. If the precision $\delta$ is small enough, then with high probability over $s$, the eigenvalue $s+\lambda$ will always be far away from any half-multiple of $\delta$ (i.e. a number of the form $(z + 0.5) \delta$, $z \in \mathbb{Z}$), for all $\lambda$. The result of phase estimation will therefore (with high probability) depend only on $\lambda$, and not on the measurement.

Using the consistent phase estimation together with Lemma~\ref{lem:HamiltonianSimulation}, we can straightforwardly derive the following lemma:
\begin{lemma}\label{lem:consistent-phaseEst-adapted}
Suppose we are given a quantum oracle that prepares copies of two unknown (normalized) $n$-qubit quantum states $\rho^+$ and $\rho^-$, and define the Hamiltonian $H = a_+ \rho^+ - a_- \rho^-$. Also assume the ranks of $\rho^+$ and $\rho^-$ are upper bounded by $r$. Then for $\delta,\xi > 0$, there is a quantum algorithm that first chooses a random shift $s$, such that with probability at least $1-2r\xi$ the following holds for all eigenstates $\ket{v_i}$ of $H$, where $H\ket{v_i} = \lambda_i\ket{v_i}$ (we call such a $s$ a \emph{good} shift):
\begin{itemize}
\item On input $\ket{v_i}\ket{\bar{0}}$, where $\ket{\bar{0}}$ is a fixed reference state, the algorithm outputs a state $O(\sqrt{\xi})$-close to $\ket{v_\lambda}\ket{f(s,\lambda_i)}$ in trace distance.
\item $f(s,\lambda_i)$ is a function only of $s$ and $\lambda_i$, and $|f(s,\lambda_i) - \lambda_i| < \delta$.
\end{itemize}
This algorithm requires $\poly(a^++a^-,\xi^{-1},\delta^{-1})$ copies of $\rho^+$ and $\rho^-$, and $\poly(n,a^++a^-,\xi^{-1},\delta^{-1})$ 1- and 2-qubit quantum gates.
\end{lemma}
This, in particular, allows us to consistently estimate eigenvalues of $K = K^+ - K^-$ using $\poly(\log n, B,\xi^{-1},\delta^{-1})$ operations in total.

For technical reasons, we will also need an approximation of the minimum eigenvalue of $K$ (possibly ignoring eigenvalues less than a threshold $\delta$). We use the following procedure:\\
\begin{algorithm}[H]
\caption{Estimation of minimum eigenvalue of $K$}
 \label{algo:lambda_min}
\begin{enumerate}
\item \textbf{Input:} Quantum oracles for $\rho^+$, $\rho^-$. A random good shift $s$ for eigenvalue estimation of $K$. Numbers $\delta, \gamma > 0$.
\item Use consistent phase estimation to estimate the eigenvalue of $K$ on $\rho^+$ and $\rho^-$, with precision $\delta$ and error probability $\xi = O(B^{-1}\delta/\log\gamma^{-1})$. Discard the estimate if its absolute value is less than $\delta$.
\item Repeat Step 2 $\Theta(B\delta^{-1}\log \gamma^{-1})$ times and output the minimum, denoted $\tilde{\lambda}_{\min}$.
\end{enumerate}
\end{algorithm}
\lem{consistent-phaseEst-adapted} implies that that with probability $1 - O(\gamma)$, the above algorithm outputs the minimum number $\tilde{\lambda}_{\text{min}}$ such that $|\tilde{\lambda}_{\text{min}}| \ge \delta$ and $\tilde{\lambda}_{\text{min}} = f(s,\lambda)$ for some eigenvalue $\lambda$ of $K$.

Finally, for operators $A$ and $B$, we will use $A\approx_{O(\epsilon)}B$ to denote that A is $O(\epsilon)$-close to B in trace distance.

\subsection{Computing the partition function}
We will prove the following lemma, using consistent phase estimation protocol:
\begin{lemma} \label{lem:app:partition}
Suppose $K = K^+ - K^-$, where $K^+$ and $K^-$ are $n \times n$ PSD matrices, and there is a quantum oracle that prepares copies of the states $\rho^+=K^+ / \tr(K^+)$, $\rho^- = K^-/ \tr(K^-)$, and an oracle for the numbers $\tr(K^+)$, $\tr(K^-)$. Moreover, assume that $\tr(K^+) + \tr(K^-) \le B$ for some bound $B$, and that $K^+$, $K^-$ have rank at most $r_{K}$. Then it is possible to estimate the partition function $Z=\tr(\exp(-K))$ to multiplicative error $\epsilon$ with success probability at least $1-\xi$, with $\poly(\log n, r_{K}, B, \epsilon^{-1}, \xi^{-1})$ quantum gates.
\end{lemma}

\begin{proof}
As stated previously, we are using consistent phase estimation to unambiguously decide whether to keep an eigenvector in our approximate support. To be precise, choose $\delta = O(\eps)$, $\xi = O(\eps^2\delta^2 B^{-2}r_{K}^{-1})$, and pick a random shift $s$ -- assume that this $s$ is a good shift (this happens with probability $1 - O(\eps^2\delta^2B^{-2})$). Define $\widetilde{Z}_{\text{supp}} = \sum_{|f(s,\lambda_i)| \ge \delta} e^{-f(s,\lambda_i)}$, and consider \algo{tildeZsupp} for estimating $\widetilde{Z}_{\text{supp}}$ (let its output be $\widetilde{X}$).
\begin{algorithm}
\caption{Estimation of $\widetilde{Z}_{\text{supp}}$}
 \label{algo:tildeZsupp}
\begin{enumerate}
\item Let $\text{sgn} = +$ with probability $\tr(K^+)/[\tr(K^+)+\tr(K^-)]$, and $\text{sgn} = -$ otherwise.
\item Use consistent phase estimation to perform eigenvalue estimation of $K$ on $\rho^{\text{sgn}}$; let the output be the normalized state $\tilde{\rho}^{\text{sgn}} = \frac{1}{\tr[K^{\text{sgn}}]}\sum_{\tilde{\lambda}} \widetilde{K}^{\text{sgn}}_{\tilde{\lambda}} \otimes \ket{\tilde{\lambda}}\bra{\tilde{\lambda}}$ for some unnormalized states $\widetilde{K}^{\text{sgn}}_{\tilde{\lambda}}$. Measure the obtained eigenvalue to obtain some $\tilde{\lambda}$. With probability at least $1-O(\xi)$, $\tilde{\lambda} = f(s,\lambda)$ for some eigenvalue $\lambda$ of $K$.
\item Compute $\tilde{\lambda}_{\min}$ by \algo{lambda_min}. If $|\tilde{\lambda}| < \delta$ or $\tilde{\lambda} < \tilde{\lambda}_{\text{min}}$ output 0; else if $\text{sgn}=+$ output $\tilde{\lambda}^{-1}e^{-\tilde{\lambda}}$; else output $-\tilde{\lambda}^{-1}e^{-\tilde{\lambda}}$.
\end{enumerate}
\end{algorithm}
In Step 3 we need to discard eigenvalues smaller than the approximate minimum eigenvalue $\tilde{\lambda}_{\text{min}}$ to keep the expectation of $\tilde{X}$ well-bounded. This is one consequence of possible error due to the application of the phase estimation procedure.

Another consequence, is that  if some eigenvalues of $K$ are close enough, they could be mapped to the same approximation $\tilde\lambda$, and are therefore treated as degenerate. We will redo the analysis to illustrate this fact: Recall that $\Pi_\lambda$ be the projection onto the $\lambda$-eigenspace of $K$. Define the projector
\begin{equation}
\widetilde{\Pi}_{\tilde{\lambda}} = \sum_{\lambda: f(s,\lambda) = \tilde{\lambda}} \Pi_\lambda
\end{equation}
to be the projector that projects onto the set of eigenvectors of $K$, with eigenvalues that get mapped to $\tilde{\lambda}$ under our consistent phase estimation procedure. We note that if we define the unnormalized states (here $\text{id}$ stands for "ideal'')
\begin{equation}
\widetilde{K}^+_{\tilde{\lambda},\text{id}} = \widetilde{\Pi}_{\tilde{\lambda}}K^+\widetilde{\Pi}_{\tilde{\lambda}}, \quad \quad
\widetilde{K}^-_{\tilde{\lambda},\text{id}} = \widetilde{\Pi}_{\tilde{\lambda}}K^-\widetilde{\Pi}_{\tilde{\lambda}}
\end{equation}
then
\begin{align}
\widetilde{K}^+_{\tilde{\lambda},\text{id}} - \widetilde{K}^-_{\tilde{\lambda},\text{id}} &= \widetilde{\Pi}_{\tilde{\lambda}}K\widetilde{\Pi}_{\tilde{\lambda}}= \sum_{i: f(s,\lambda_i) = \tilde{\lambda}} \lambda_i \ket{v_i}\bra{v_i} \approx_{\xi} \tilde{\lambda} \sum_{i: f(s,\lambda_i) = \tilde{\lambda}} \ket{v_i}\bra{v_i}  = \tilde{\lambda} \widetilde{\Pi}_{\tilde{\lambda}}.
\end{align}
Note that consistent phase estimation implements an operation $O(\sqrt{\xi})$-close to the operation $\sum_{\tilde{\lambda}}\widetilde{\Pi_{\lambda}} \otimes \ket{\tilde{\lambda}}\bra{\tilde{\lambda}}$, and therefore
\begin{align}
\tilde{\rho}^{\text{sgn}}{\equiv} \frac{1}{\tr[K^{\text{sgn}}]}\sum_{\tilde{\lambda}} \widetilde{K}^{\text{sgn}}_{\tilde{\lambda}} \otimes \ket{\tilde{\lambda}}\bra{\tilde{\lambda}}&\approx_{O(\sqrt{\xi})} \frac{1}{\tr[K^{\text{sgn}}]} \sum_{\tilde{\lambda}} \widetilde{\Pi}_{\tilde{\lambda}} K^{\text{sgn}} \widetilde{\Pi}_{\tilde{\lambda}} \otimes \ket{\tilde{\lambda}}\bra{\tilde{\lambda}} \nonumber \\
&=\frac{1}{\tr[K^{\text{sgn}}]}\sum_{\tilde{\lambda}} \widetilde{K}^{\text{sgn}}_{\tilde{\lambda},\text{id}} \otimes \ket{\tilde{\lambda}}\bra{\tilde{\lambda}}.
\end{align}
Thus $\tilde{K}^{\text{sgn}}_{\tilde{\lambda}} \approx_{O(B\sqrt{\xi})} \tilde{K}^{\text{sgn}}_{\tilde{\lambda},\text{id}}$, and hence
\begin{equation}
\tilde{K}^+_{\tilde{\lambda}} - \tilde{K}^-_{\tilde{\lambda}} \approx_{O(B\sqrt{\xi})} \tilde{\lambda}\widetilde{\Pi}_{\tilde{\lambda}}.
\end{equation}
We see that the consistent phase estimation of Step 2 serves to approximately project $\rho^+$ or $\rho^-$ onto the span of eigenvectors of $K$ with eigenvalue approximately equal to some $\tilde{\lambda}$; and on this space, the unnormalized output states at Step 2 approximately differ only on by a multiple of the identity on their support. There is therefore a basis of vectors $\{\ket{\tilde{v}_i}\}$ where $\widetilde{K}^+_{\tilde{\lambda}}$ and $\widetilde{K}^-_{\tilde{\lambda}}$ are approximately diagonal for all $\tilde{\lambda}$. These vectors are approximate eigenvectors of $K$, i.e.
\begin{equation}\label{eq:tilde-v}
\|K \ket{\tilde{v}_i} - \tilde{\lambda}_i \ket{\tilde{v}_i} \| = O(\xi)
\end{equation}
for some numbers $\tilde{\lambda}_i$.\footnote{Note that the basis $\{\ket{\tilde{v}_i}\}$ and exact eigenbasis of $K$, $\{\ket{v_i}\}$, are not necessarily equivalent, because the vectors in the former are only approximate eigenvectors of $K$.} Working in the approximate eigenbasis basis, we can write
\begin{equation}
\widetilde{K}^+_{\tilde{\lambda}} \approx_{O(B\sqrt{\xi})} \sum_{i: \tilde{\lambda}_i = \tilde{\lambda}} \tilde{\lambda}^+_i \ket{\tilde{v}_i}\bra{\tilde{v}_i}, \quad \quad
\widetilde{K}^-_{\tilde{\lambda}} \approx_{O(B\sqrt{\xi})} \sum_{i: \tilde{\lambda}_i = \tilde{\lambda}} \tilde{\lambda}^-_i \ket{\tilde{v}_i}\bra{\tilde{v}_i}
\end{equation}
for nonnegative numbers $\tilde{\lambda}^+_i - \tilde{\lambda}^-_i = \tilde{\lambda}$. This gives the following approximation for $\tilde{\rho}^\text{sgn}$:
\begin{equation}\label{eq:tilde-rho-app}
\tilde{\rho}^{\text{sgn}} \approx_{O(\sqrt{\xi})} \frac{1}{\tr[K^{\text{sgn}}]}\sum_i \tilde{\lambda}^{\text{sgn}}_i \ket{\tilde{v}_i}\bra{\tilde{v}_i} \otimes \ket{\tilde{\lambda}_i}\bra{\tilde{\lambda}_i}.
\end{equation}
Therefore the expectation of $\widetilde{X}$ is upper bounded by
\begin{align}
\E[\widetilde{X}] &\le \frac{\sum_i e^{-\tilde{\lambda}_i}}{\tr[K^+]+\tr[K^-]} + O(\sqrt{\xi}) \max_{i: \tilde{\lambda}_i \ge \tilde{\lambda}_{\text{min}}} \tilde{\lambda}_i^{-1}e^{-\tilde{\lambda}} \\
&\le (1+\delta) \frac{\widetilde{Z}_{\text{supp}}}{\tr[K^+]+\tr[K^-]} + O(\sqrt{\xi}\delta^{-1})\widetilde{Z}_{\text{supp}} \\
&\le (1 + \delta + \sqrt{\xi}\delta^{-1}B) \frac{\tilde{Z}_{\text{supp}}}{\tr[K^+]+\tr[K^-]} \\
&= (1+O(\epsilon))\frac{\tilde{Z}_{\text{supp}}}{\tr[K^+]+\tr[K^-]}.
\end{align}
A similar bound holds for lower bounding $\E[\widetilde{X}]$, showing that knowing $\E[\widetilde{X}]$ would give a $O(\epsilon)$-multiplicative error approximation to $\widetilde{Z}_{\text{supp}}$. Just as in the ideal case, we can simply repeat our procedure $O(\eta^{-1}B^2\delta^{-2}\epsilon^{-2})$ times and take the mean to obtain a $O(\epsilon)$-multiplicative error approximation of $\E[\widetilde{X}]$, and hence of $\tilde{Z}_{\text{supp}}$.

As before, we also need to estimate the number of eigenvalues $\lambda$ (including degeneracy) with $|f(s,\lambda)| < \delta$, i.e. the number of $i$'s with $|\tilde{\lambda}_i| < \delta$. Let this number be $n - \tilde{R}$. Let the output of the following procedure be $\tilde{Y}$:\\
\begin{algorithm}[H]
\caption{Estimation of $n-\tilde{R}$}
 \label{algo:n-tildeR}
\begin{enumerate}
\item Perform consistent phase estimation to estimate eigenvalues of $K$ on the uniformly random state $I/n$; let the output eigenvalue be $\tilde{\lambda}$.
\item If $|\tilde{\lambda}| < \delta$ output 1; otherwise output 0.
\end{enumerate}
\end{algorithm}
It is clear that $n\E[\tilde{Y}]$ is an $O(r_{K}\sqrt{\xi})$-multiplicative error approximation of $n-\tilde{R}$, and it can be proven as before that  $O(r_{K}\epsilon^{-2})$ repetitions of the above procedure suffice to give an $O(\epsilon)$-error multiplicative stimate of $(n-\tilde{R})/n$. It can again be argued that $O(r_{K}\epsilon^{-2})$ repetitions suffice to estimate $\E[\widetilde{Y}]$, and thus $n-\widetilde{R}$, to $O(\epsilon)$-multiplicative error.

Finally, to estimate the full partition function we merely note that $\widetilde{Z}_{\text{supp}} + (n-\widetilde{R})$ is an $O(\delta+\epsilon) = O(\epsilon)$-multiplicative error estimate of the partition function $Z$; we can therefore estimate $Z$ by estimating both terms separately and taking the sum.
\end{proof}

\subsection{Computing the Gibbs function}\label{append:Gibbs-function}
In this section we prove the following result:
\begin{theorem}\label{thm:gibbs}
Suppose $K = K^+ - K^-$, where $K^+$ and $K^-$ are $n \times n$ PSD matrices, and there is a quantum oracle that prepares copies of the states $\rho^+=K^+ / \tr(K^+)$, $\rho^- = K^-/ \tr(K^-)$, and an oracle for the numbers $\tr(K^+)$, $\tr(K^-)$. Moreover, assume that $\tr(K^+) + \tr(K^-) \le B$ for some bound $B$, and that $K^+$, $K^-$ have rank at most $r_{K}$. Then it is possible to prepare the Gibbs state $\rho_G=\exp(-K)/\tr(\exp(-K))$ up to error $\epsilon$ in trace distance, with $\poly(\log n, r_{K}, B, \epsilon^{-1})$ quantum gates.
\end{theorem}
\begin{proof}
The procedure will the sketch given in \append{low-rank-Gibbs}, but using consistent phase estimation rather than the na{\"i}ve protocol. We again assume we chose a good shift $s$ for the operator $K$, and first consider a procedure to sample from
\begin{equation}
\tilde{\rho}_{\text{supp}} \equiv \sum_{i:|\tilde{\lambda}_i| \ge \delta} e^{-\tilde{\lambda}_i}\ket{\tilde{v}_i}\bra{\tilde{v}_i}/\widetilde{Z}_{\text{supp}},
\end{equation}
where $\delta > 0$ is a small threshold and $\tilde{\lambda}_i$, $\ket{\tilde{v}_i}$ were defined previous in~\eqref{eq:tilde-v}.  $\rho_{\text{supp}}$ is the Gibbs state when considering only the space spanned by approximate eigenvectors of $K$ whose eigenvalues estimates (under consistent phase estimation) are at least $\delta$ in aboslute value. Again choose $\delta = O(\eps)$, $\xi = O(\eps^2\delta^2 B^{-2}r_{K}^{-1})$, and pick a good random shift $s$ -- assume that this $s$  (this happens with probability $1 - O(\eps^2\delta^2B^{-2})$). Consider \algo{tilderhosupp} below.

\begin{algorithm}[htbp]
\caption{Estimation of $\tilde{\rho}_{\text{supp}}$}
 \label{algo:tilderhosupp}
\begin{enumerate}
\item \textbf{Input:} A good random shift $s$, an $O(\epsilon)$-multiplicative error estimate $\widetilde{Z}'_{\text{supp}}$ of $\widetilde{Z}_{\text{supp}}$, quantum oracles for $\rho^+$, $\rho^-$.
\item Let $\text{sgn} = +$ with probability $\tr(K^+)/[\tr(K^+)+\tr(K^-)]$, and $\text{sgn} = -$ otherwise.
\item Use consistent phase estimation to perform eigenvalue estimation of $K$ on $\rho^{\text{sgn}}$; let the output be the normalized state $\tilde{\rho}^{\text{sgn}} = \frac{1}{\tr[K^{\text{sgn}}]}\sum_{\tilde{\lambda}} \widetilde{K}^{\text{sgn}}_{\tilde{\lambda}} \otimes \ket{\tilde{\lambda}}\bra{\tilde{\lambda}}$ for unnormalized states $\widetilde{K}^{\text{sgn}}_{\tilde{\lambda}}$. Including the randomness on choosing $\text{sgn}$, we have the state
\begin{align}
\tilde{\rho} &\equiv \frac{\tr[K^+]}{\tr[K^+]+\tr[K^-]}\tilde{\rho}^+ + \frac{\tr[K^-]}{\tr[K^+]+\tr[K^-]}\tilde{\rho}^- \\
&\approx_{O(\sqrt{\xi})} \frac{1}{\tr[K^+]+\tr[K^-]} \sum_i (\tilde{\lambda}^+_i + \tilde{\lambda}^-_i)\ket{\tilde{v}_i}\bra{\tilde{v}_i} \otimes \ket{\tilde{\lambda}_i}\bra{\tilde{\lambda}_i}.
\end{align}
Here in the second line we used the approximation~\eqref{eq:tilde-rho-app}.
\item Apply the projection $I \otimes \sum_{|\tilde{\lambda}| < \delta, \tilde{\lambda} \ge \tilde{\lambda}_{\text{min}}} \ket{\tilde{\lambda}}\bra{\tilde{\lambda}}$, where $\tilde{\lambda}_{\min}$ is the output of \algo{lambda_min}. In other words, measure the second register to make sure that $|\tilde{\lambda}| \ge \delta$ and $\tilde{\lambda} \ge \tilde{\lambda}_{\text{min}}$, and reject otherwise.
\item Apply the measurement operator
\begin{equation}
\sum_{i: |\tilde{\lambda}_i| \ge \delta, \tilde{\lambda}_i \ge \tilde{\lambda}_{\text{min}}} \frac{\delta}{\tilde{\lambda}^+ + \tilde{\lambda}^-}\frac{e^{-\tilde{\lambda}}}{2Z'_{\text{supp}}} \ket{\tilde{v}_i}\bra{\tilde{v_i}} \otimes \ket{\tilde{\lambda}_i}\bra{\tilde{\lambda}_i},
\end{equation}
up to $O(\sqrt{\xi})$ error, to the state. We can do this by first estimating $\tilde{\lambda}^+ + \tilde{\lambda}^-$ by quantum principal analysis (i.e. phase estimation of $\tilde{\rho}$), to precision $O(\sqrt{\xi}\delta)$ and error probability $O(\sqrt{\xi})$; then accept the resulting state with probability approximately $\frac{\delta}{\tilde{\lambda}^+ + \tilde{\lambda}^-}\frac{e^{-\tilde{\lambda}}}{2Z'_{\text{supp}}}$.
\end{enumerate}
\end{algorithm}
\algo{tilderhosupp} will give us a good approximation for $\tilde{\rho}_{\text{supp}}$, the Gibbs state on the approximate support of $K$ (ignoring small eigenvalues). As before, we will need to approximate the Gibb state on the approximate kernel of K as well, which we define as
\begin{equation}
\tilde{\rho}_{\text{ker}} = \frac{1}{n-\widetilde{R}} \sum_{i: \tilde{\lambda}_i} \ket{v_i}\bra{v_i}.
\end{equation}
This state can easily be approximated by starting with the completely mixed state $I/n$ and performing consistent phase estimation to estimate eigenvalues of $K$, postselecting on the case that the measured estimate is less than $\delta$ in magnitude.

To complete our estimation for the full Gibbs state, we see that $\rho_G = \exp(-K)/\tr(\exp(-K))$ can be approximated by
\begin{align}
\rho_G &\approx_{O(\delta)} \frac{1}{Z} \sum_{i} e^{-\tilde{\lambda}_i} \ket{\tilde{v}_i}\bra{\tilde{v}_i} \\
&= \frac{\widetilde{Z}_{\text{supp}}}{Z} \frac{1}{\widetilde{Z}_{\text{supp}}} \sum_{i:|\tilde{\lambda}_i| \ge \delta} e^{-\tilde{\lambda}_i} \ket{\tilde{v}_i}\bra{\tilde{v}_i} + \frac{n-\widetilde{R}}{Z} \frac{1}{n-\widetilde{R}} \sum_{i:|\tilde{\lambda}_i| < \delta} e^{-\tilde{\lambda}_i} \ket{\tilde{v}_i}\bra{\tilde{v}_i}\\
&\approx_{O(\epsilon)} \frac{\widetilde{Z}_{\text{supp}}}{Z} \tilde{\rho}_{\text{supp}} + \frac{n-\widetilde{R}}{Z} \tilde{\rho}_{\text{ker}}.
\end{align}
Thus by Lemma 7 of \cite{vanApeldoorn2017quantum} , it suffices to have $O(\epsilon)$-multiplicative error estimates for $\widetilde{Z}_{\text{supp}}$, $n-\widetilde{R}$, and $Z$, and $O(\epsilon)$-trace distance error approximations for $\tilde{\rho}_{\text{supp}}$ and $\tilde{\rho}_{\text{ker}}$. We have already shown how to achieve all of this.
\end{proof}

\end{document}